%% file: main.tex
\documentclass[11pt]{article}
\usepackage[margin=1in]{geometry}

\usepackage[utf8]{inputenc} 
\usepackage{booktabs}       
\usepackage{amsfonts}       
\usepackage{nicefrac}       
\usepackage{microtype}      
\usepackage[svgnames]{xcolor}
\usepackage{parskip}        
\usepackage{amsmath, amsthm, mathtools, dsfont}
\usepackage{txfonts} 
\usepackage{color,graphicx}
\usepackage{url,hyperref}
\usepackage{enumerate}
\usepackage{algorithm}
\usepackage[noend]{algpseudocode}
\usepackage{fancybox}
\usepackage{xspace}
\usepackage{cuted,tcolorbox,lipsum}
\usepackage{multicol}
\usepackage{xcolor}
\usepackage[T1]{fontenc}
\usepackage{tikz}
\usepackage{tcolorbox}
\usepackage{hyperref}
\usepackage[all]{hypcap}
\usetikzlibrary{calc}
\tcbuselibrary{skins}
\usepackage[colorinlistoftodos]{todonotes}
\usepackage{array,multirow}
\usetikzlibrary{decorations.pathreplacing}

\usepackage{color-edits}
\addauthor{kb}{cyan}
\addauthor{ms}{blue}
\addauthor{dc}{red}
\addauthor{ig}{orange}
\addauthor{mm}{purple}
\usepackage{csquotes}
\usepackage{subfigure}

\input{macros.tex}

\title{Bi-Criteria Metric Distortion}

\author{ 
Kiarash Banihashem\thanks{Computer Science Department, University of Maryland, College Park. {\tt kiarash@umd.edu}.}
\and
Diptarka Chakraborty\thanks{School of Computing,
National University of Singapore. {\tt diptarka@comp.nus.edu.sg}.}
\and 
Shayan Chashm Jahan\thanks{Computer Science Department, University of Maryland, College Park. {\tt schjahan@umd.edu}.}
\and 
Iman Gholami\thanks{Computer Science Department, University of Maryland, College Park. {\tt igholami@umd.edu}.}
\and 
MohammadTaghi Hajiaghayi\thanks{Computer Science Department, University of Maryland, College Park. {\tt hajiagha@umd.edu}.}
\and 
Mohammad Mahdavi\thanks{Computer Science Department, University of Maryland, College Park. {\tt mahdavi@umd.edu}.}
\and 
Max Springer\thanks{Mathematics Department, University of Maryland, College Park. {\tt mss423@umd.edu}.}
}
\date{}
\begin{document}

\maketitle
\begin{abstract}
Selecting representatives based on voters' preferences is a fundamental problem in social choice theory. While cardinal utility functions offer a detailed representation of preferences, ordinal rankings are often the only available information due to their simplicity and practical constraints. The metric distortion framework addresses this issue by modeling voters and candidates as points in a metric space, with distortion quantifying the efficiency loss from relying solely on ordinal rankings. Existing works define the cost of a voter with respect to a candidate as their distance and set the overall cost as either the sum (utilitarian) or maximum (egalitarian) of these costs across all voters. They show that deterministic algorithms achieve a best-possible distortion of 3 for any metric when considering a single candidate.

This paper explores whether one can obtain a better approximation compared to an optimal candidate by relying on a committee of $k$ candidates ($k \ge 1$), where the cost of a voter is defined as its distance to the closest candidate in the committee.
We answer this affirmatively in the case of line metrics, demonstrating that with $O(1)$ candidates, it is possible to achieve optimal cost. Our results extend to both utilitarian and egalitarian objectives, providing new upper bounds for the problem. We complement our results with lower bounds for both the line and 2-D Euclidean metrics.
\end{abstract}

\input{intro}

\section{Preliminaries}
\label{sec:prelim}

\paragraph{Metric space.} Let us consider a domain $\mathcal{X}$ and a distance function $d :  \mathcal{X} \times \mathcal{X} \to  \mathbb{R}$. We call $\left(\mathcal{X},d \right)$ a \emph{metric space} if the distance function $d$ satisfies the following properties:  
\begin{itemize}  
    \item \textbf{Positive definite}: For all $x,y \in \mathcal{X}$, $d(x,y) \ge 0$, and $d(x,y) = 0$ iff $x=y$.
    \item \textbf{Symmetry}: For all $x,y \in \mathcal{X}$, $d(x,y) = d(y,x)$.  
    \item \textbf{Triangle inequality}: For all $x,y,z \in \mathcal{X}$, $d(x,y) \le d(x,z) + d(z,y)$.
\end{itemize}  
In this paper, we consider 
\begin{itemize}
    \item \textbf{Line metric (1-D Euclidean metric)}: The domain is $\mathcal{X} = \mathbb{R}$, and for any two points $p,q \in \mathbb{R}$, their distance is $d(p,q) = |p - q|$.
    \item \textbf{2-D Euclidean metric}: The domain is $\mathcal{X} = \mathbb{R}^2$, and for any two points $p=(p_x,p_y),q=(q_x,q_y) \in \mathbb{R}^2$, their distance is $d(p,q) = \left|\left|p-q\right|\right|_2:= \left( \left(p_x-q_x \right)^2 + \left(p_y-q_y \right)^2\right)^{1/2}$.
\end{itemize}

\paragraph{Election.} An \emph{election instance} $\elec = (V, C, \succ)$ consists of a set $V = \{v_1, \dots, v_n\}$ of $n$ voters and a set $C = \{c_1, \dots, c_m\}$ of $m$ candidates. Each voter $v_i \in V$ has a linear order $\GE{i}$ over the candidates, where $c_j \GE{i} c_k$ indicates that voter $v_i$ prefers $c_j$ over $c_k$. We refer to $\GE{i}$ as the \emph{ordinal preference} of voter $v_i$. Furthermore, $\succ \hspace{0.2em} = \{\GE{1}, \dots, \GE{n}\}$ is called the \emph{preference profile} of the voters. 
Additionally, the $j$-th candidate in the ordinal preference of voter $v_i$ is denoted by $\GE{i,j}$.

We consider the voters and candidates to lie in the same metric space $(\mathcal{X},d)$. For ease of exposition, we extend the notion of distance $d$ to be defined directly on the set of voters and candidates instead of the points they occupy in the underlying space $\mathcal{X}$. We say a (distance) metric $d$ is \emph{consistent} with an election instance $\elec= (V, C, \succ)$, denoted as $d \vartriangleright \elec$, when for any voter $v_i$, $c_j \GE{i} c_k$ if $d(v_i, c_j) \le d(v_i, c_k)$. 

\paragraph{Social cost.} Let us consider an election instance $\elec = (V, C, \succ)$, and a distance metric $d \vartriangleright \elec$. Let $I = (\elec, d)$ denote the instance $\elec$ with $d$ being its underlying distance metric. For any subset of candidates $S \subseteq C$, and a voter $v \in V$, we use $d(v,S)$ to denote the distance between the voter $v$ to its nearest neighbor in $S$, i.e., $d(v,S):=\min_{c \in S} d(v,c)$. In this paper, we focus on the following two \emph{social costs}:
\begin{itemize}
    \item \textbf{Sum-cost (Utilitarian objective)}: For any subset of candidates $S \subseteq C$, its \emph{sum-cost}, denoted by $\cost_s(S,I)$ is defined as $\cost_s(S,I):= \sum_{v \in V}d(v,S)$.
    \item \textbf{Max-cost (Egalitarian objective)}: For any subset of candidates $S \subseteq C$, its \emph{max-cost}, denoted by $\cost_m(S,I)$ is defined as $\cost_m(S,I):= \max_{v \in V}d(v,S)$.
\end{itemize}
When it is clear from the context, we drop $I$ and simply use $\cost_s(S)$ and $\cost_m(S)$.

\paragraph{Voting rule and distortion.} A (deterministic) \emph{voting rule} (also referred to as \emph{mechanism}) $f$ is a function that maps an election instance $\elec$ to a subset of candidates $S$. We use algorithms and mechanisms interchangeably throughout this paper. 
In this paper, we are interested in comparing the cost of voting rules that selects $k$-committees with the cost of an optimal single candidate.
We call a single candidate $\copt$ \emph{optimal} if
$$\cost(\copt) = \min_{c\in C}\cost(c).$$
Throughout the paper we use $\opt = \cost(\copt)$ to refer to the cost of an optimal single candidate.

To capture how good a voting rule is in the worst case, the notion of distortion is used. For any voting rule $f$, its \emph{distortion}, or more specifically, \emph{1-distortion} is defined as
\[
\dist(f):=\sup_{\elec}\sup_{d \vartriangleright \elec}\frac{\cost(f(\elec))}{\opt}
\]
where the cost function $\cost$ in the above definition could be either $\cost_s$ or $\cost_m$ depending on the context. In other words, the 1-distortion compares the cost of the mechanism to the cost of an optimal candidate in the worst case.

\section{Line metric election: the order of candidates and voters}
\label{order}

In this section, we present an algorithm designed to determine the order of candidates and voters for any line metric election instance $\elec = (V, C, \succ)$ based on the voters' preference profile. This step is essential for establishing the upper bounds discussed in Section \ref{sc:upper-bounds}. Moreover, this approach may prove useful for future research, as it offers a general method for obtaining the total order of candidates and voters as explained in the following.

This algorithm focuses on a specific subset of candidates with properties useful for the purposes of this paper and potential future work on line metric distortion.  

\begin{definition}[\Core{}]  
In an election $\elec = (V, C, \succ)$, a subset $A \subseteq C$ is called \emph{a \core{}} if for any voter $v_i$, any candidate $c_j \in A$, and any candidate $c_k \in C \setminus A$, we have $c_j \GE{i} c_k$.  
\end{definition}  

The algorithm introduced in this section determines the order only for a subset $\Det$, referred to as the \emph{determined candidates}, rather than for all candidates. It expands $\Det$ iteratively and establishes the order of candidates within $\Det$. By the end of the algorithm, $\Det$ forms a \core{} subset of candidates.  

Additionally, for some voters with similar preferences, it may be impossible to distinguish whether they are on the right or left; thus, they are considered to be at the same point. These properties of the voter and candidate order retrieved by the algorithm are accounted for in the analysis and other sections of this paper. Overall, obtaining the order of voters and candidates with this level of accuracy remains significant for the intended purposes.

The algorithm consists of three parts. The first part, called \splitline, involves dividing the line into two halves and determining whether each candidate occurs on the left or right side (some of them remain undetermined). The second part, denoted as \sortcandidates, finds the order of candidates in each half and merges the sorted lists to obtain a total ordering. The final part, referred to as \sortvoters, determines the voters' ordering based on the retrieved order of candidates. Algorithm \ref{alg:sort} demonstrates how these three components contribute to sorting the candidates and voters.

\begin{algorithm}[ht]
  \caption{Sort candidates and voters}
  \label{alg:sort}
  \hspace*{\algorithmicindent} \textbf{Input:} Election instance $\elec$.\\
  \hspace*{\algorithmicindent} \textbf{Output:} Sequence $S_C$ where is the order of determined candidates and sequence $S_V$ the order of voters.
  \begin{algorithmic}[1]
    \Function{\sortcandiatesandvoters}{$\elec = (V, C, \succ)$}
        \State $(L, R) \gets \splitline(\elec)$
        \label{line:split-line}
        \State $S_C \gets \sortcandidates(\GE{1}, L, R)$
        \label{line:sort-cands}
        \State $S_V \gets \sortvoters(\elec, S_C)$
        \label{line:sort-vots}
        \State \Return $(S_C, S_V)$
    \EndFunction
  \end{algorithmic}
\end{algorithm}

\paragraph{\splitline.} In this part, we first find a pivot to split the line at that point. To achieve this, we introduce the following definitions. We arbitrarily choose the first voter as the pivot voter, associated with two pivot candidates as follows. 

\begin{definition}[Pivot Voter and Candidates]  
\label{def:pivot-candidates}
The voter $v_1$ is called the \emph{pivot voter}. Additionally, the two nearest candidates to the pivot voter are called the \emph{pivot candidates}. Without loss of generality, assume that $c_1$ and $c_2$ are the two nearest candidates to the pivot voter $v_1$, with $c_1$ positioned to the left of $c_2$. Note that $c_1$ and $c_2$ may both be on the same side of $v_1$.
\end{definition}

Finally, the point where the line is split is defined as follows:  

\begin{definition}[Pivot Point]  
\label{def:pivot}
The midpoint of the line segment between the two pivot candidates is called the \emph{pivot point}, denoted by $p$. Let $L$ and $R$ be the subsets of candidates on the left and right sides of $p$, respectively.  
\end{definition}

Finally, We aim to determine whether a candidate belongs to $L$ or $R$. Therefore, we formally define determined and undetermined candidates as follows:

\begin{definition}  
    \label{def:determine}  
    A candidate is \emph{determined} if it is known whether the candidate belongs to $L$ or $R$ based on Definition \ref{def:pivot}. Otherwise, the candidate is \emph{undetermined}. Let $\Det = L \cup R$ be the set of determined candidates
\end{definition}

Initially, based on Definition \ref{def:pivot-candidates}, we know that $c_1 \in L$ and $c_2 \in R$. Thus, $\Det = \{c_1, c_2\}$. Through several iterations, we expand $L$ and $R$ by adding as many candidates as possible. Therefore, at the end of each iteration we update $\Det$ such that $\Det = L \cup R$ again. We also ensure that $\Det$ always forms a consecutive subset of candidates on the line.
The process for determining new candidates in each iteration is as follows:

If there exists a voter $v_i$ such that two candidates $c_k$ and $c_j$ satisfy $c_k \GE{i} c_j$, where $c_j \in \Det$ but $c_k \notin \Det$, then we can determine $c_k$'s membership as follows:
\begin{itemize}
    \item
If both $c_1$ and $c_2$ are closer to $v_i$ than $c_j$, then $c_k$ belongs to the opposite side of $c_j$.
    \item 
Otherwise, $c_k$ belongs to the same side as the candidate in $\{c_1, c_2\}$ that is closer to $v_i$.
\end{itemize}

Algorithm \ref{alg:determine-one} is pseudocode for the function \determine, which determines a candidate if the above conditions hold and adds it to the corresponding set. Algorithm \ref{alg:split-line} determines as many candidates as possible in each iteration while maintaining the succession of the determined candidates (see \ref{sec:order-analysis} for proofs). We call these candidates $\New$. At the end of the iteration, it merges $\New$ into $\Det$.
It is important to note that we do not add each point immediately to $\Det$; instead, we merge them all at the end. This approach ensures that $\Det$ remains a consecutive list of candidates, which is crucial for our analysis.

The output of this function is $L$ and $R$, which represent the determined candidates on the left and right sides of the pivot point $p$, respectively.

\begin{algorithm}[ht]
  \caption{Determine a candidate with a determined candidate in a voter's ordinal preference}
  \label{alg:determine-one}
  \hspace*{\algorithmicindent} {\textbf{Input:} Ordinal preference of voter $v_i$, denoted $\GE{i}$; candidates $c_j$ and $c_k$ where $c_k \GE{i} c_j$, $c_j$ is determined but $c_k$ is not; two sets $L$ and $R$ containing currently determined candidates on the left and right sides of the $p$, respectively; and pivot candidates, $c_1$ and $c_2$.}
  
  \hspace*{\algorithmicindent} {\textbf{Output:} Updated sets $L$ and $R$ including candidate $c_k$.}
  \begin{algorithmic}[1]
    \Function{\determine}{$\GE{i}, c_j, c_k, L, R, c_1, c_2$}
        \If{$c_1 \GE{i} c_k \land c_2 \GE{i} c_k$}
            \State Add $c_k$ to $L$ if $c_j \in R$, otherwise add $c_k$ to $R$
        \Else
            \State Add $c_k$ to $L$ if $c_1 \GE{i} c_2$, otherwise add $c_k$ to $R$
        \EndIf
        \State \Return $(L, R)$
    \EndFunction
  \end{algorithmic}
\end{algorithm}

\begin{algorithm}[ht]
  \caption{Determine Candidates}
  \label{alg:split-line}
  \hspace*{\algorithmicindent} \textbf{Input:} Election instance $\elec$.\\
  \hspace*{\algorithmicindent} \textbf{Output:} Two subsets $L$ and $R$ of candidates, where candidates are on the left or right of the pivot point $p$.
  \begin{algorithmic}[1]
    \Function{\splitline}{$\elec = (V, C, \succ)$}
        \State $L \gets \{c_1\}$
        \State $R \gets \{c_2\}$
        \State $\Det \gets \{c_1, c_2\}$
        \Repeat
        \label{line:iteration}
            \State $\New \gets \emptyset$
            \While{$\exists (v_i, c_j, c_k): c_k \GE{i} c_j \land c_k \notin \Det \land c_j \in Det$}
                \State $(L, R) \gets \determine(\GE{i}, c_k, c_j, L, R, c_1, c_2)$
                \State $\New \gets \New \cup \{c_k\}$
            \EndWhile
            \State $\Det \gets \Det \cup \New$
            \label{line:update-det}
        \Until{$\New = \emptyset$}
        \State \Return $(L, R)$
    \EndFunction
  \end{algorithmic}
\end{algorithm}

\paragraph{\sortcandidates.} The goal of this part is to sort the candidates in $L$ and $R$. We know that $L$ lies to the left of $p$, and $R$ lies to the right of $p$, where $p$ is the midpoint of the segment connecting the two nearest candidates of $v_1$ (see Definitions \ref{def:pivot-candidates} and \ref{def:pivot}). 
Consequently, in the ordinal preference of $v_1$, candidates in $L$ with higher preferences are positioned to the right of those with lower preferences. Similarly, candidates in $R$ with higher preferences are positioned to the left of those with lower preferences. By combining these two observations, we can sort the candidates based on their positions along the line. Algorithm \ref{alg:sort-candidates} presents the pseudocode for this approach.

\begin{algorithm}[ht]
  \caption{Sort determined candidates}
  \label{alg:sort-candidates}
  \hspace*{\algorithmicindent} \textbf{Input:} Preference order of $v_1$, denoted as $\GE{i}$, $L$ and $R$, candidates on the left and the right side of pivot $p$.  
  \hspace*{\algorithmicindent} \textbf{Output:} Sequence $S_C$, sorted candidates in $L$ and $R$ by their position left to right.
  \begin{algorithmic}[1]
    \Function{\sortcandidates}{$\GE{1}, L, R$}
        \State $S_C$ is an empty sequence of candidates.
        \For {$c_i$ in $\GE{1}$}
            \If{$c_i \in L$}
                \State Add $c_i$ to the left of $S_C$.
            \ElsIf{$c_i \in R$}
                \State Add $c_i$ to the right of $S_C$.
            \EndIf
        \EndFor
        \State \Return $S_C$
    \EndFunction
  \end{algorithmic}
\end{algorithm}

\paragraph{\sortvoters.} This part focuses on sorting voters given the sorted determined candidates. The key observation is that a voter who prefers one candidate over another tends to be closer to the candidate they prefer more. Consequently, we have a method to compare two voters. For a pair of voters $v_i$ and $v_j$, assume $k$ is the smallest index where the ordinal preferences of $v_i$ and $v_j$ differ. $v_i$ is on the left side of $v_j$ if $\GE{i,k}$ is on the left side of $\GE{j,k}$. If $\GE{i,k}$ and $\GE{j,k}$ are not determined, assuming $v_i$ and $v_j$ are at the same position does not affect the further analysis (see Sections \ref{sec:order-analysis}).  

It is important to note that the voter sorting process, denoted as $\sortvoters(\elec = (V, C, \succ), S_C)$, takes the election and the sorted determined candidates as input and outputs $S_V$, a sequence of all voters in the order determined by the above approach.

In the next subsection, we provide a detailed proof demonstrating why the mentioned algorithm functions correctly.

\subsection{Analysis}
\label{sec:order-analysis}

In this part, we focus on showing that we correctly calculated the order of candidates and voters. As mentioned earlier, this approach does not determine the exact order of all candidates or accurately rank all voters. First, recall that the algorithm returns the sequence $S_C$ as the order of candidates (Line \ref{line:sort-cands}). This sequence contains the set of determined candidates, denoted as $\Det = L \cup R$. Additionally, in Line \ref{line:sort-vots} of Algorithm \ref{alg:sort}, we calculate the order of voters based on $S_C$. As previously noted, this sequence is not entirely accurate due to limited available information. The following definition formalizes this order.

\begin{definition}
    For a subset of candidates $A \subseteq C$, the order of voters \emph{with respect to} $A$ is a sequence $S$ of voters from left to right, such that there is a tie between voters $v_i$ and $v_j$ if and only if their ordinal preferences over the candidates in $A$ are identical.
\end{definition}

Now, we can formally introduce the main theorem implied by the mentioned algorithm:

\begin{theorem} 
\label{thm:order}
Let $(\mathbb{R}, d)$ be the line metric. For an election $\elec = (V, C, \succ)$ such that $d \vartriangleright \elec$, \sortcandiatesandvoters{} correctly identifies:
\begin{enumerate}
    \item the subset $\Det = L \cup R$ which is a \core{} subset of candidates,
    \item the order of $\Det$ candidates, denoted as $S_C$,
    \item and the order of voters with respect to $\Det$, denoted as $S_V$.
\end{enumerate}
\end{theorem}

For the remainder of this section, all lemmas consider a specific line metric election instance. 

The following lemma demonstrates that a prefix of a voter's ordinal preference forms a consecutive subsequence of candidates.

\begin{lemma}
\label{lm:consecutive-prefix}
For any voter $v_i \in V$ and $1 \le k \le m$, the $k$ most preferred candidates of voter $v_i$ form a consecutive subsequence of candidates.
\end{lemma}

\begin{proof}
    Assume the condition does not hold for a voter $v_i$ and some $k$. 
    Since any single candidate forms a consecutive subsequence, we must have $k \ge 2$. 
    Let $A$ denote the $k$ most preferred candidates of $v_i$. By assumption, there exist two candidates $c_l, c_r \in A$ such that a candidate $c_m \in C \setminus A$ lies between them.  

    Assume without loss of generality that $c_l$ is to the left of $c_m$ and $c_r$ is to the right of $c_m$.
    Further, suppose $v_i$ is located to the left of $c_m$. Since $c_m \notin A$, we must have $c_r \GE{v_i} c_m$. However, since the candidates are arranged in a line, it follows that $c_m \GE{v_i} c_r$, which is a contradiction (illustrated in Figure \ref{fig:consecutive}).
    \begin{figure}[ht]
        \centering
        \begin{tikzpicture}
            \draw[thick] (-5, 0) -- (5, 0);
            
            \draw[blue, fill=blue] (-2, 0) circle (0.1) node[above] {$c_l$};
            \draw[blue, fill=blue] (3, 0) circle (0.1) node[above] {$c_r$};

            \draw[black, fill=black] (0, 0) circle (0.1) node[above] {$c_m$};
            
            \draw[red, thick] (-1.5, 0) node[below] {$v$};
            \draw[red, thick]  (-1.5, 0) node {$\times$};
        \end{tikzpicture}
        \caption{This figure is the illustration of the succession of nearest of any voter.}
        \label{fig:consecutive}
    \end{figure}
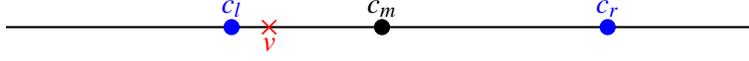
\end{proof}

According to Lemma \ref{lm:consecutive-prefix}, we have:
\begin{corollary}
\label{cl:c1-c2}
    $c_1$ and $c_2$ are consecutive.
\end{corollary}

Next, recalling the definitions of determined and undetermined candidates (\ref{def:determine}), the following lemma formally proves that the function \determine{} in Algorithm \ref{alg:determine-one} correctly determines an undetermined candidate.

\begin{lemma}
\label{lm:extend}
    Assume that $\Det$ is a consecutive subset of determined candidates, including $c_1$ and $c_2$. If there exists a voter $v_i$ and candidates $c_j \in \Det$ and $c_k \notin \Det$ such that $c_k \GE{i} c_j$, then the function \determine{} correctly determines $c_k$.
\end{lemma}

\begin{proof}
Let us consider two cases regarding the positioning of $c_1$, $c_2$, and $c_k$ in the ordinal preference of $v_i$.

Case 1: Both $c_1$ and $c_2$ are positioned before $c_k$ in the ordinal preference of $v_i$.

Without loss of generality, we can assume that $c_j$ is in $R$. We use proof by contradiction to show that $c_k$ is in $L$. Assume that $c_i$ is in $R$. Consider the prefix of the ordinal preference of $v_i$ ending with $c_k$. By Lemma \ref{lm:consecutive-prefix}, they must form a consecutive set of candidates. Therefore, $c_j$ must be on the right side of $c_k$.  However, since $\Det$ consists of consecutive candidates, $c_k$ would necessarily be positioned on the right side of $c_j$. 
Because of the contradiction we can conclude $c_k$ is in $L$ (Illustrated in Figure \ref{fig:case1})

\begin{figure}[ht]
        \centering
        \begin{tikzpicture}
            \draw[thick] (-5, 0) node[left] {\redcross} -- (5, 0);
            \draw[thick] (0, -0.1) -- (0, 0.1) -- (0,0) node[above] {$p$};
            
            \draw[blue, fill=blue] (-1, 0) circle (0.1) node[above] {$c_1$};
            \draw[blue, fill=blue] (1, 0) circle (0.1) node[above] {$c_2$};
            \draw[blue, fill=blue] (4, 0) circle (0.1) node[above] {$c_j$};

            \draw[red] (-1, -0.2) -- (-1, -0.5) -- (4, -0.5) -- (4, -0.2);
            \draw[red] (1, -0.2) -- (1, -0.5);

            \draw[Green] (-1, 0.5) -- (-1, 0.8) -- (2.5, 0.8) -- (2.5, 0.5);
            \draw[Green] (1, 0.5) -- (1, 0.8);

            \draw[black, fill=black] (2.5, 0) circle (0.1) node[above] {$c_k$};
        \end{tikzpicture}
        \\[0.5cm]
        \begin{tikzpicture}
            \draw[thick] (-5, 0) node[left] {\greencheck} -- (5, 0);
            \draw[thick] (0, -0.1) -- (0, 0.1) -- (0,0) node[above] {$p$};
            
            \draw[blue, fill=blue] (-1, 0) circle (0.1) node[above] {$c_1$};
            \draw[blue, fill=blue] (1, 0) circle (0.1) node[above] {$c_2$};
            \draw[blue, fill=blue] (4, 0) circle (0.1) node[above] {$c_j$};

            \draw[red] (-1, -0.2) -- (-1, -0.5) -- (4, -0.5) -- (4, -0.2);
            \draw[red] (1, -0.2) -- (1, -0.5);

            \draw[Green] (1, 0.5) -- (1, 0.8) -- (-2.5, 0.8) -- (-2.5, 0.5);
            \draw[Green] (-1, 0.5) -- (-1, 0.8);

            \draw[black, fill=black] (-2.5, 0) circle (0.1) node[above] {$c_k$};
        \end{tikzpicture}
        \caption{For a voter $v_i$, if we have $c_1 \GE{i} c_k$, $c_2 \GE{i} c_k$, 
        and $c_k \GE{i} c_k$, then $c_k$ and $c_j$ cannot both be in the same side of $p$. The figure above illustrates this contradiction, while the one below shows that they can be on opposite sides.}
        \label{fig:case1}
\end{figure}
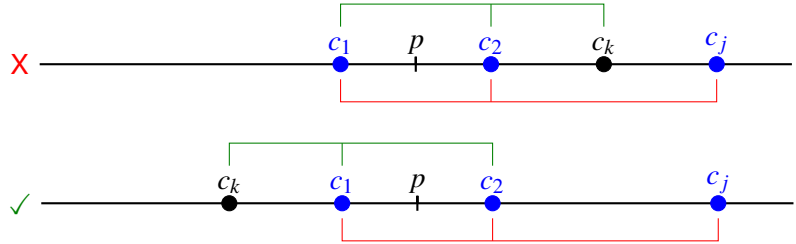

Case 2: At least one of $c_1$ and $c_2$ is positioned after $c_k$ in the ordinal preference of $v_i$.

Without loss of generality, assume that $c_1$ precedes $c_2$ in the ordinal preference of $v_i$. Consider the prefix of the ordinal preference of $v_i$ that contains $c_1$ and $c_k$ but not $c_2$. By Lemma~\ref{lm:consecutive-prefix}, this prefix must form a consecutive sequence of candidates, meaning $c_2$ cannot lie between $c_1$ and $c_k$.
If $c_k$ were in $R$, $c_2$ would lie between $c_1$ and $c_k$, as $c_1$ and $c_2$ are consecutive.
This contradiction implies that $c_k$ is in $L$ (Illustrated in Figure \ref{fig:case2}).

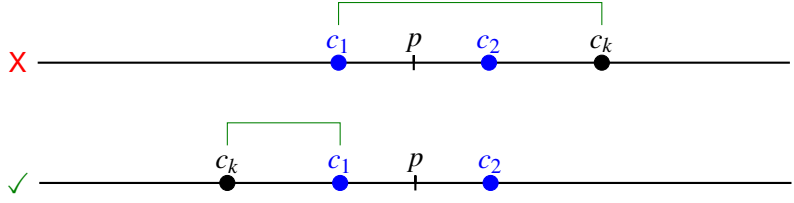
\begin{figure}[ht]
        \centering
        \begin{tikzpicture}
            \draw[thick] (-5, 0) node[left] {\redcross} -- (5, 0);
            \draw[thick] (0, -0.1) -- (0, 0.1) -- (0,0) node[above] {$p$};
            
            \draw[blue, fill=blue] (-1, 0) circle (0.1) node[above] {$c_1$};
            \draw[blue, fill=blue] (1, 0) circle (0.1) node[above] {$c_2$};

            \draw[Green] (-1, 0.5) -- (-1, 0.8) -- (2.5, 0.8) -- (2.5, 0.5);

            \draw[black, fill=black] (2.5, 0) circle (0.1) node[above] {$c_k$};
        \end{tikzpicture}
        \\[0.5cm]
        \begin{tikzpicture}
            \draw[thick] (-5, 0) node[left] {\greencheck} -- (5, 0);
            \draw[thick] (0, -0.1) -- (0, 0.1) -- (0,0) node[above] {$p$};
            
            \draw[blue, fill=blue] (-1, 0) circle (0.1) node[above] {$c_1$};
            \draw[blue, fill=blue] (1, 0) circle (0.1) node[above] {$c_2$};

            \draw[Green] (-1, 0.5) -- (-1, 0.8) -- (-2.5, 0.8) -- (-2.5, 0.5);

            \draw[black, fill=black] (-2.5, 0) circle (0.1) node[above] {$c_k$};
        \end{tikzpicture}
        \caption{For a voter $v_i$, if we have $c_1 \GE{i} c_2$ 
        and $c_k \GE{i} c_2$, then $c_k$ is in $L$. The figure above illustrates that if $c_k$ were in $R$, then $c_2$ would be in the consecutive subsequence of $c_1$ and $c_k$ which is a contradiction.
        On the other hand, the one below shows that $c_k$ must be in $L$.}
        \label{fig:case2}
\end{figure}
\end{proof}

As explained in Lemma \ref{lm:extend}, determining a new candidate works if the current set of $\Det$ is consecutive. Therefore, in the following Lemma we show that in the beginineg of each iteration, $\Det$ is a consecutive subsequence of candidates.

\begin{lemma}
\label{lm:consecutive-determined}
    In Algorithm \ref{alg:split-line}, in the beginning of each iteration (Line \ref{line:iteration}), all determined candidates, denoted as $\Det$, form a consecutive subsequence of all candidates.
\end{lemma}
\begin{proof}
    Let $\Det_i$ be the set of determined candidates in the beginning of $i$-th iteration. 
    We prove that for any $i$, $\Det_i$ is a consecutive subsequence. $\Det_1 = \{c_1, c_2\}$ satisfies the condition. (Corollary \ref{cl:c1-c2}).
    
    In the $(i-1)$-th iteration ($i \ge 2$), For any voter $v_j$, let $\last_j$ be the least preferred candidate of $v_j$ in $\Det_{i-1}$. Then by Lemma \ref{lm:extend}, by the end of iteration $i-1$ any candidate $c_k$ such that $c_k \GE{v} \last_j$ is determined. We call this set of candidates $\prefix_j$; note that $\Det_{i-1} \subseteq \prefix_j$.  As $\prefix_j$ is a prefix of the ordinal preference of $v_j$, by Lemma \ref{lm:consecutive-prefix}, $\prefix_j$ is a consecutive subsequence of candidates. Also $\prefix_j$ contains $c_1$ and $c_2$ because $\{c_1, c_2\} \subseteq \Det_{i-1} \subseteq \prefix_j$.
   
    Since $\Det_i = \Det_{i-1} \cup \New$ (Line \ref{line:update-det}) and $\New = \bigcup_{j} (\prefix_j \setminus \Det_{i-1})$, we have $\Det_i = \bigcup_{j} P_j$.
    Observe however that a union of intersecting intervals is always an interval; this can be seen from an easy inductive argument. Since $\prefix_j$ all contain $\{c_1, c_2\}$, it follows that $\Det_i$ is a consecutive subsequence as well. 
\end{proof}

Next, we show that $\Det$ is a \core{} subset of candidates.

\begin{lemma}
    \label{lm:det-cand-opt}
    By the end of the algorithm, $\Det$ is \core{}. i.e.
    For any voter $v_i$, candidate $c_j\in\Det$, $c_k\notin\Det$, we have $c_j \GE{i} c_k$.
\end{lemma}
\begin{proof}
    Assume otherwise that $c_k \GE{i} c_i$. Then, we can determine $c_k$ using $c_j$ and voter $v_i$ according to Lemma \ref{lm:extend}, and consequently, Algorithm \ref{alg:split-line} cannot have terminated at this point. Therefore, $c_j \GE{i} c_i$ for any voter $v_i$.
\end{proof}

We now demonstrate that the function \sortcandidates{} in Algorithm \ref{alg:sort-candidates} returns the determined candidates in sorted order from left to right.

\begin{lemma}
\label{lm:candidates-order}
    Function \sortcandidates{} correctly sorts determined candidates, denoted as $\Det = L \cup R$ from left to right.
\end{lemma}
\begin{proof}
Recall that sets $L$ and $R$ contain determined candidates (Definition \ref{def:determine}). Now, consider the ordinal preference of $v_1$. For any pair of candidates $c_i, c_j \in L$, we show that $d(c_i, c_1) < d(c_j, c_1)$ if $c_i \GE{1} c_j$.

There are two cases:

\begin{itemize}
    \item $c_1$ is on the right side of $v_1$; we have:
    \begin{align*}
        d(c_i, c_1) - d(c_j, c_1) &= d(c_i, v_1) - d(c_1, v_1) - d(c_j, v_1) + d(c_1, v_1) \\
        &= d(c_i, v_1) - d(c_j, v_1)
    \end{align*}
    
    \item $c_1$ is on the left side of $v_1$; we have:
    \begin{align*}
        d(c_i, c_1) - d(c_j, c_1) &= d(c_i, v_1) + d(c_1, v_1) - d(c_j, v_1) - d(c_1, v_1) \\
        &= d(c_i, v_1) - d(c_j, v_1)
    \end{align*}
\end{itemize}

Therefore, we have:
\[
d(c_i, c_1) < d(c_j, c_1) \Leftrightarrow d(c_i, v_1) < d(c_j, v_1)
\]

Based on the definition of ordinal preference, we have $d(c_i, v_1) < d(c_j, v_1)$ if $c_i \GE{1} c_j$. Consequently $d(c_i, c_1) < d(c_j, c_1)$ if $c_i \GE{1} c_j$.

Similarly, for any pair of candidates $c_i, c_j \in R$, we can show that $d(c_i, c_2) < d(c_j, c_2)$ if  $c_i \GE{1} c_j$.

The function \sortcandidates{} processes candidates according to $\GE{1}$, adding them to the left of $S_C$ if they are in $L$ and to the right otherwise. Therefore, $S_C$ contains candidates in the correct order from left to right.
\end{proof}

Next, we show the correctness of function \sortvoters{}.

\begin{lemma}
\label{lm:voters-order}
    function \sortvoters{} 
    finds the order of voters with respect to determined candidates.
\end{lemma}
\begin{proof}
    By Lemma \ref{lm:candidates-order}, we have the order of a determined candidates. Assume that for candidate $c_i$, $\order_{c_i}$ is the index of $c_i$ from left to right. As \splitline{} determines a prefix for each candidate, for a voter $v_i$, we let sequence $\sorted_j = \order_{\GE{i,1}}, \order_{\GE{i,2}}, \dots \order_{\GE{i,k}}$ where $k$ is the number of determined candidates. Now, we can compare voters' $\sorted$ sequences lexicographically. A voter with smaller $\sorted$ is on the left side of a voter with larger $\sorted$. It is important to highlight that this approach may result in ties among some voters in the ordering.
\end{proof}

By Lemmas \ref{lm:det-cand-opt}, \ref{lm:candidates-order}, and \ref{lm:voters-order}, we have Theorem \ref{thm:order} proved.

\section{Line metric upper bounds}
\label{sc:upper-bounds}

\subsection{Sum-cost objective}
\label{sc:sum-upper-bound}

In this section, we study the sum objective, where the cost of a set of candidates is defined as the sum of the distance of voters to the candidate set. We consider the line metric, where voters and candidates are located on the real line, and utilize the orderings obtained in Section \ref{order} to present upper bounds on distortion. We start by showing that it is possible to choose a set of three candidates, such that an optimal candidate is always chosen. We then improve this by showing that we can always omit one of the three selected candidates, resulting in a shortlist of size two that includes an optimal candidate. We note that this not only shows a distortion of two based on our distortion metric, but it also ensures that the actual optimal single candidate is included in the selected list. 

First, we identify the optimal single candidate when candidates and voters are positioned on the real line. In the following proofs, we often refer to the median voter. When the number of voters is even, we can consider either of the middle two candidates as the median.

\begin{lemma}
    \label{lm:real-line-opt}
    When considering the sum of distances objective and the candidates and voters are located on the real line, one of the candidates directly to the right or left of the median voter will be an optimal candidate.
\end{lemma}
\begin{proof}
Assume, without loss of generality, that there exists a candidate $c_l$ who is the first candidate to the left of the median voter $v$. Consider any candidate $c$ to the left of $c_l$. As we move from $c$ to $c_l$, all voters to the right of $v$ (including $v$) will experience a \textit{reduction} in their distance to the candidates by the difference between $c$ and $c_l$, while the distance for other voters can increase by at most the same amount. Since $v$ is the median voter, there are at least as many voters to the right of $v$ (including $v$) as there are to its left (excluding $v$). Therefore, the total cost cannot increase when moving from $c$ to $c_l$, implying that the cost of $c_l$ is \textit{at most} the cost of $c$. Similarly, if a candidate $c_r$ directly to the right of $v$ exists, any candidate farther to the right will have a distance that is at least as great. Hence, either $c_l$ or $c_r$ will be an optimal candidate that minimizes the total cost.
\end{proof}

Next, we show that we can select three candidates to ensure an optimal candidate, which is immediately to the left or right of the median voter is selected. 

\begin{lemma}
\label{lm:three-cand}
There exists a voting rule for $3$-committee election on line metric such that the resulting committee includes an optimal candidate.
\end{lemma}
\begin{proof}
First, we can determine the order of a subset of candidates $\Det$ and all voters based on the algorithms in Section \ref{order}. 
Additionally, by Lemma \ref{lm:det-cand-opt}, any candidate $c \not\in \Det$ cannot be an optimal candidate, since $c$ will be farther from all voters than any candidate in $\Det$. 
Thus, we can ignore the candidates outside $\Det$ and proceed with the remaining candidates.
    
Now, we consider the median voter $v$ in the ordering of voters. 
We note that there might be a tie between multiple voters for this position, but we will only utilize the voter's closest candidate, which will be the same for all tied voters. 
    
After identifying the median voter, we consider $v$'s closest candidate $c$. 
This candidate will be either the one immediately to the left or the right of the median voter $v$. 
Next, we use the ordering of the candidates to find the candidates $c_l$ to the left of $c$ and $c_r$ to the right of $c$. 
We claim that the set $\{c,c_l,c_r\}$ is guaranteed to include an optimal candidate: 
if $c$ is to the left of the median voter, then $c_r$ will be the first candidate to the right of the median, and one of $c$ or $c_r$ will be an optimal candidate by Lemma \ref{lm:real-line-opt}. 
Similarly, if $c$ is to the right of the median, one of $c$ or $c_l$ will be an optimal candidate. 
Therefore, the set $\{c,c_l,c_r\}$ will contain an optimal candidate. 

Finally, we note that candidates $c_l$ and $c_r$ might not exist. 
In such cases, the candidates to the left and right of $v$ are still selected, if they exist.

\end{proof}

Next, we show that we can omit either $c_l$ or $c_r$, choosing only two candidates, while still including an optimal candidate. 

\begin{theorem}
    There exists a voting rule for 2-committee election on line metric such that the resulting committee includes an optimal candidate.
\end{theorem}
\begin{proof}
By Lemma \ref{lm:three-cand}, we can select three candidates that include an optimal candidate. 
Let $c_1$, $c_2$, and $c_3$ denote these candidates from left to right. 
We then define the sets $V_1$, $V_2$, and $V_3$ as the sets of voters who prefer the corresponding candidate to the other two, as illustrated in Figure \ref{fig:2cands}. 
Thus, we can state that
\begin{align*}
    \cost_s(c_2) &= \sum_{v \in V} d(v,c_2) \\
    &= \sum_{v \in V_1} d(v,c_2) + \sum_{v \in V_2} d(v,c_2) + \sum_{v \in V_3} d(v,c_2)\\
    &\leq \sum_{v \in V_1} d(v,c_2) + \sum_{v \in V_2} d(v,c_2) + \sum_{v \in V_3} \left(d(v,c_3) + d(c_2,c_3) \right) \tag{Triangle Inequality}\\
    &\leq \sum_{v \in V_1} d(v,c_2) + \sum_{v \in V_2} d(v,c_3) + \sum_{v \in V_3} \left(d(v,c_3) + d(c_2,c_3)\right) \tag{$\forall_{v\in V_2}\ d(v,c_2) \leq d(v,c_3)$}\\
    &= \sum_{v \in V_1} \left(d(v,c_3) - d(c_3,c_2)\right) + \sum_{v \in V_2} d(v,c_3) + \sum_{v \in V_3} \left(d(v,c_3) + d(c_2,c_3)\right)\tag{$\forall_{v\in V_1}\ d(v,c_2) = d(v,c_3) - d(c_2,c_3)$}\\
    &=\cost_s(c_3) - \lvert V_1 \rvert \cdot d(c_2,c_3) + \lvert V_3 \rvert \cdot d(c_2,c_3) \\
    &= \cost_s(c_3) + (\lvert V_3 \rvert - \lvert V_1 \rvert) \cdot d(c_2,c_3).
\end{align*}
Similarly, we can show that $\cost_s(c_2) \leq \cost_s(c_1) + (\lvert V_1 \rvert - \lvert V_3 \rvert) \cdot d(c_2,c_1)$. Now, depending on whether $\lvert V_1 \rvert < \lvert V_3 \rvert$ or not, we can see that either 
 $\cost_s(c_2) \leq \cost_s(c_1)$ or $\cost_s(c_2) \leq \cost_s(c_3)$. So we can disregard one of $c_1$ and $c_3$ based on this and still find the optimum, as $c_2$ is guaranteed to be a better candidate than the discarded one. 
 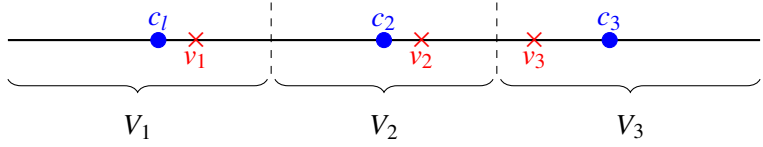
\begin{figure}[t]
        \centering
        \begin{tikzpicture}
            \draw[thick] (-5, 0) -- (5, 0);
            
            \draw[blue, fill=blue] (-3, 0) circle (0.1) node[above] {$c_l$};
            \draw[blue, fill=blue] (0, 0) circle (0.1) node[above] {$c_2$};
            \draw[blue, fill=blue] (3, 0) circle (0.1) node[above] {$c_3$};
            
            \draw[red, thick] (-2.5, 0) node[below] {$v_1$};
            \draw[red, thick]  (-2.5, 0) node {$\times$};
            \draw[red, thick] (.5, 0) node[below] {$v_2$};
            \draw[red, thick]  (.5, 0) node {$\times$};
            \draw[red, thick] (2, 0) node[below] {$v_3$};
            \draw[red, thick]  (2, 0) node {$\times$};

            \draw[dashed]  (-1.5, 0.5) -- (-1.5, -0.5) ;
            \draw[dashed]  (1.5, 0.5) -- (1.5, -0.5) ;

            \draw [decorate,decoration={brace,amplitude=5pt,mirror,raise=3ex}]
                (-5,0) -- (-1.55,0) node[midway,yshift=-3em]{$V_1$};
            \draw [decorate,decoration={brace,amplitude=5pt,mirror,raise=3ex}]
                (-1.45,0) -- (1.45,0) node[midway,yshift=-3em]{$V_2$};
            \draw [decorate,decoration={brace,amplitude=5pt,mirror,raise=3ex}]
                 (1.55,0) -- (5,0) node[midway,yshift=-3em]{$V_3$};
        \end{tikzpicture}
        \caption{A figure illustrating three candidates $c_1$, $c_2$, and $c_3$ along with the possible locations of voters closest to each candidate, $V_1$, $V_2$, and $V_3$. Voters $v_1$, $v_2$ and $v_3$ show examples of voters in each set.}
        \label{fig:2cands}
    \end{figure}
\end{proof}

\begin{corollary}
    There exists a voting rule for $2$-committee election on line metric with $1$-distortion of $1$.
\end{corollary}

 Finally, we show that for a general metric, it is possible to achieve a $1$-distortion of $1+\frac{2}{m-1}$ when choosing $k=m-1$ candidates. 
\begin{theorem}
    \label{thm:up-m-1}
    There exists a voting rule choosing $m-1$ out of $m$ candidates achieving a $1$-distortion of $1+\frac{2}{m-1}$ for the sum-cost objective.
\end{theorem}
\begin{proof}
    For each voter $v$, let $\first_v$ be their closest candidate. Then, we claim that the voting rule that chooses the $m-1$ candidates appearing most frequently in $\first$ achieves the desired distortion. For a given instance, let $C$ be the set of candidates selected by this algorithm and $c_{\text{opt}}$ be the optimal single candidate. If $c_{\text{opt}}\in C$, then we get a $1$-distortion of $1$ and we are done. Otherwise, $C$ includes every candidate except for $c_{\text{opt}}$. Now, we can bound the optimal cost $\opt$ in this instance as 
\begin{align*}
    \opt &= \sum_{i\in[n]} d(\copt,v) \\&= \sum_{\substack{v\in V\\ \first_v=\copt}} d(\copt,v) + \sum_{\substack{v\in V\\ \first_v\neq \copt}} d(\copt,v)\\
    &\geq \sum_{\substack{v\in V\\ \first_v=\copt}} d(\copt,v) + \sum_{\substack{v\in V\\ \first_v\neq \copt}} d(C,v). \tag{$\first_v \in C$ if $\first_v\neq o$}
\end{align*}
Let $v'$ be the voter closest to $\copt$ such that $\first_v \neq o$. Then, we can use this to state that 
\begin{align*}
\opt &\geq \sum_{\substack{v\in V\\ \first_v\neq \copt}} d(\copt,v) \\
&\geq \sum_{\substack{v\in V\\ \first_v\neq \copt}} d(\copt,v') \\
&\geq \frac{m-1}{m}\cdot n \cdot d(\copt,v') \tag{$\first_v=\copt$ for at most $\frac{n}{m}$ voters based on choice of $C$}
\end{align*}
and therefore 
\begin{equation} \label{lm:mid-m-1}
\frac{n}{m} d(\copt,v') \leq \frac{1}{m-1}\opt.
\end{equation}
Now, we can bound the cost of $C$ as follows:
\begin{align*}
    \cost_m(C) &= \sum_{v\in V} d(C,v)\\
    &= \sum_{\substack{v\in V\\ \first_v=\copt}} d(C,v) + \sum_{\substack{v\in V\\ \first_v\neq \copt}} d(C,v)\\
    &\leq \sum_{\substack{v\in V\\ \first_v=\copt}} (d(\copt,v) + d(\copt,v') + d(C,v')) + \sum_{\substack{v\in V\\ \first_v\neq \copt}} d(C,v) \tag{Triangle inequality}\\
    &\leq \sum_{\substack{v\in V\\ \first_v=\copt}} (d(\copt,v) + d(\copt,v') + d(\copt,v')) + \sum_{\substack{v\in V\\ \first_v\neq o}} d(C,v) \tag{$\first_{v'}\ne \copt$ and $\first_{v'}\in C$}\\
    &\leq \opt + \sum_{\substack{v\in V\\ \first_v=\copt}} 2d(\copt,v')\tag{Previous upper bound for OPT}\\
    &=\opt + 2\cdot \frac{n}{m} d(\copt,v') \tag{$\first_v=\copt$ for at most $\frac{n}{m}$ voters}\\
    &\leq(1 + \frac{2}{m-1}) \opt \tag{By Equation \ref{lm:mid-m-1}}
\end{align*}
\end{proof}

\subsection{Max-cost objective}

Next, we focus on the max objective, where the cost of a set of candidates is the maximum distance of any voter from the set. We show that selecting four candidates is sufficient for achieving a distortion of $1$ compared to the single optimal candidate. In addition, we show that we can achieve distortion factors of $3/2$ and $2$ by selecting three or two candidates respectively. We note that in different parts of this section, we refer to the leftmost and rightmost voters $v_l$ and $v_r$ which we identify using the algorithms in Section \ref{order}. While the orderings in that section may have ties, we are only concerned with the most preferred candidate of each voter, which will be the same for tied voters.

We begin by determining a lower bound for the optimal cost $\opt$ when only one candidate is selected.

\begin{lemma}
\label{lm:distance-cand-vl-vr}
Let $v_l$ and $v_r$ be the leftmost and rightmost candidates respectively. If the distance between \( v_l \) and \( v_r \) is \( d \) (i.e., \( d = d(v_l, v_r) \)), then the cost for the optimal single candidate $\opt$ satisfies \( \opt \geq \frac{d}{2} \).
\end{lemma}

\begin{proof}  
Let \( c_{\text{opt}} \) be an optimal candidate. The distance between \( v_l \) and \( c_{\text{opt}} \) is \( d(v_l, c_{\text{opt}}) \), and the distance between \( v_r \) and \( c_{\text{opt}} \) is \( d(v_r, c_{\text{opt}}) \). By the triangle inequality:
\[
d(v_l, c_{\text{opt}}) + d(v_r, c_{\text{opt}}) \geq d.
\]
Thus,
\[
\max(d(v_l, c_{\text{opt}}), d(v_r, c_{\text{opt}})) \geq \frac{d}{2}.
\]
Therefore, \( \opt \geq \frac{d}{2} \).

\end{proof}

Based on the algorithms in Section \ref{order}, we can determine the order of candidates in $\Det$ and the voters.
We denote the leftmost voter as \( v_l \) and the rightmost voter as \( v_r \). Similarly, we denote the closest candidate to \( v_l \) as \( c_l \) and the closest candidate to \( v_r \) as \( c_r \).

Now, we show that if there are no candidates between \( v_l \) and \( v_r \), a distortion of $1$ can be achieved with 2 candidates. 

\begin{lemma}
If there are no candidates placed between \( v_l \) and \( v_r \),  a distortion of $1$ can be achieved by selecting \( c_l \) and \( c_r \).
\end{lemma}

\begin{proof}
In this case, the voters are next to each other in a block, with \( c_l \) being the first candidate immediately to the left of all voters and \( c_r \) being the first candidate immediately to the right of all voters. Therefore, by selecting \( c_l \) and \( c_r \), we ensure that the closest candidate to each voter is included.
Therefore, the cost of this set is at most that of the single optimal candidate, and a distortion factor of $1$ is achieved compared to this benchmark.
\end{proof}

From now on, we assume at least one candidate is located between $v_l$ and $v_r$. This implies that the optimal single candidate $c_{\text{opt}}$ is also between $v_l$ and $v_r$, as any candidate outside the interval has a cost of at least $d(v_l,v_r)$.

Next, we show that selecting two candidates is sufficient to achieve a distortion of $2$ compared to the optimal single candidate. 

We first prove the following lemma for the left-most and right-most voters, denoted \( v_l \) and \( v_r \). 

\begin{lemma}
\label{lm:distance-voter-vl-vr}

If the distance between \( v_l \) and \( v_r \) is \( d \), then for every voter \( v \), either \( d(v, v_l) \leq \opt \) or \( d(v, v_r) \leq \opt \).
\end{lemma}

\begin{proof} 
Since \( v \) is a voter between \( v_l \) and \( v_r \) (recall \( v_l \) is the left-most and \( v_r \) is the right-most voter), we have:
\[
d(v_l, v) + d(v, v_r) = d.
\]
Thus, either \( d(v_l, v) \leq \frac{d}{2} \) or \( d(v, v_r) \leq \frac{d}{2} \).

By Lemma \ref{lm:distance-cand-vl-vr}, \( \frac{d}{2} \leq \opt \). Therefore, we conclude:
\[
d(v_l, v) \leq \frac{d}{2} \leq \opt, \quad \text{or} \quad d(v, v_r) \leq \frac{d}{2} \leq \opt.
\]
\end{proof}

Now, we show that it's possible to select two candidates and achieve a distortion of $2$.
\begin{theorem}
    \label{thm:max-up-2}
    Let $v_l$ and $v_r$ be the leftmost and rightmost voters, and $c_l$ and $c_r$ the closest candidates to $v_l$ and $v_r$. Then, set $\{c_l,c_r\}$ has a cost at most twice the optimal single candidate.
\end{theorem}
\begin{proof}
Since \( c_l \) is closest to \( v_l \) and we know \( d(v_l, c_{\text{opt}}) \leq \opt \), it follows that:
\[
d(v_l, c_l) \leq d(v_l, c_{\text{opt}})\leq  \opt.
\]
Similarly, \( d(v_r, c_r) \leq \opt \).

For any voter \( v \), by Lemma \ref{lm:distance-voter-vl-vr}, either \( d(v, v_l) \leq \opt\) or \( d(v, v_r) \leq \opt \). Without loss of generality, assume that \( d(v, v_l) \leq \opt \). Then:
\[
d(v, c_l) \leq d(v, v_l) + d(v_l, c_l) \leq \opt + \opt = 2 \cdot \opt.
\]
Similarly, in the other case:
\[
d(v, c_r) \leq 2 \cdot \opt.
\]

Thus, for every voter, the distance to the closest candidate in $\{c_l,c_r\}$ is at most \( 2 \cdot \opt \). Hence, this selection achieves a distortion of $2$.
    
\end{proof}

Figure \ref{fig:opt4} illustrates how the closest candidates to $v_l$ and $v_r$ lying outside the interval between them leads to a distortion factor larger than $1$. 
Specifically, if $c_l$ is to the left of $v_l$ or $c_r$ is to the right of $v_r$, some voters between $v_l$ and $v_r$ might be further than $\opt$ from both $c_l$ and $c_r$.

\begin{figure}[t]
    \centering
    \begin{tikzpicture}
        \draw[thick] (-5, 0) -- (5, 0);
        
        \draw[blue, fill=blue] (-3, 0) circle (0.1) node[below] {$c_l$};
        \draw[blue, fill=blue] (3, 0) circle (0.1) node[below] {$c_r$};

        \draw[red, thick] (-2, 0) node[below] {$v_l$};
        \draw[red, thick] (-2, 0) node {$\times$};
        \draw[red, thick] (2, 0) node[below] {$v_r$};
        \draw[red, thick] (2, 0) node {$\times$};
        
        \draw[<->, thick] (-3, 0.5) -- (-2, 0.5) node[midway, above] {$\leq \textit{OPT}$};
        \draw[<->, thick] (2, 0.5) -- (3, 0.5) node[midway, above] {$\leq \textit{OPT}$};
    \end{tikzpicture}
    \caption{Illustration of the candidates $c_l$ and $c_r$, and the voters $v_l$ and $v_r$, demonstrating why choosing these candidates does not achieve a distortion of $1$}
    \label{fig:opt4}
\end{figure}
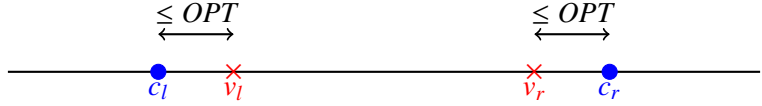

On the other hand, if the candidates $c_l$ and $c_r$ are placed between $v_l$ and $v_r$, as shown in Figure \ref{fig:opt4-swapped}, this selection would achieve a distortion of $1$. In this case, the following property holds for any voter $v$:
\[
d(v, c_l) \leq \max(d(v, v_l),\opt) \quad \text{and} \quad d(v, c_r) \leq \max(d(v, v_r),\opt)
\]
as $c_l$ lies between $v_l$ and $c_{\text{opt}}$ and $c_r$ lies between $v_r$ and $c_{\text{opt}}$.

By Lemma \ref{lm:distance-voter-vl-vr}, we know that for every voter $v$, either $d(v, v_l) \leq \opt$ or $d(v, v_r) \leq \opt$. Combining these, it follows that:
\[
d(v, c_l) \leq \textit{OPT} \quad \text{or} \quad d(v, c_r) \leq \textit{OPT}.
\]

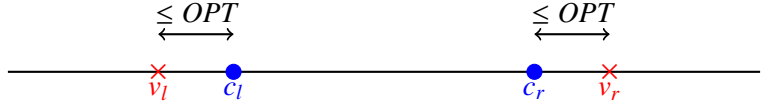
\begin{figure}[t]
    \centering
    \begin{tikzpicture}
        \draw[thick] (-5, 0) -- (5, 0);
        
        \draw[blue, fill=blue] (-2, 0) circle (0.1) node[below] {$c_l$};
        \draw[blue, fill=blue] (2, 0) circle (0.1) node[below] {$c_r$};

        \draw[red, thick] (-3, 0) node[below] {$v_l$};
        \draw[red, thick] (-3, 0) node {$\times$};
        \draw[red, thick] (3, 0) node[below] {$v_r$};
        \draw[red, thick] (3, 0) node {$\times$};
        
        \draw[<->, thick] (-3, 0.5) -- (-2, 0.5) node[midway, above] {$\leq \textit{OPT}$};
        \draw[<->, thick] (2, 0.5) -- (3, 0.5) node[midway, above] {$\leq \textit{OPT}$};
    \end{tikzpicture}
    \caption{Illustration of the candidates $c_l$ and $c_r$ positioned between the voters $v_l$ and $v_r$. In this case, we achieve the optimal answer instead of a two approximation.}
    \label{fig:opt4-swapped}
\end{figure}

To improve the distortion factor, we attempt to include the leftmost and rightmost candidates that lie between the first and last voters by selecting more candidates. 
We first demonstrate that introducing an additional candidate allows us to achieve a distortion of $3/2$ and then show that selecting two additional candidates guarantees a distortion of $1$.

\begin{theorem}
    There exists a voting rule to select three candidates that achieves a distortion of $3/2$ compared to the optimal single candidate.
\end{theorem}
\begin{proof}
    Let $v_l$ and $v_r$ be the leftmost and rightmost voters, and $c_l$ and $c_r$ their closest candidates. Using the ordering of the candidates, let $c_r'$ be the candidate immediately to the left of $c_r$. We claim that the set $\{c_l,c_r,c_r'\}$ achieves a distortion of $3/2$. First, we note that any voter to the left of $c_l$ or the right of $c_r'$ has their closest candidate included in this set. Now, one of $c_r'$ or $c_r$ must be the rightmost candidate between $v_l$ and $v_r$. Let this be candidate $c_r^*$. We already know that any voter outside the interval between $c_l$ and $c_r^*$ has their closest candidate included. In addition, we have
    \begin{align*}
        d(c_l,c_r^*)&\leq d(c_l,v_l)+d(v_l,c_{\text{opt}})+d(c_{\text{opt}},c_r^*)\leq 3\opt.
    \end{align*}
    Therefore, the minimum distance of any voter in this interval to endpoints $c_l,c_r^*$ is at most $3\opt/2$. So, the maximum distance of any voter to the candidates is at most  $3/2\opt$ and we achieve a distortion of $3/2$.
\end{proof}

Finally, we show that selecting four candidates allows us to achieve a distortion of $1$.

\begin{theorem}
    \label{thm:max-up-4}
    There exists a voting rule to select four candidates that achieves a distortion of $1$ compared to the optimal single candidate.
\end{theorem}
\begin{proof}
    We again consider the leftmost and rightmost voters $v_l$ and $v_r$ and their closest candidates $c_l$ and $c_r$. In addition, we choose candidate $c_l'$ to the right of $c_l$ and $c_r'$ to the left of $c_r$ based on the ordering of candidates. Now, one of $c_l$ and $c_l'$ will be the first candidate to the right of $v_l$: if $c_l$ is to the right of $v_l$, it must be the first such candidate. Otherwise, it is the first candidate to the left of $v_l$, so $c_l'$ is to the right of $v_l$. Let this candidate be $c_l^*$ and define $c_r^*$ similarly. Now, any voter $v$ to the left of $c_l'$ has their closest candidate included in the set, as $v$ is either between $c_l$ and $c_l'$, so one of these two is $v$'s closest candidate, or $v$ is to the left of $c_l$, in which case $c_l$ must be $v$'s closest candidate. Similarly, the closest candidate to each voter to the right of $c_r'$ is selected. 
    
    Next, since both $c_l^*$ and $c_r^*$ are between $v_l$ and $v_r$ and $d(v_l,v_r)\leq2\opt$, we have $d(c_l^*,c_r^*)\leq2\opt$. Therefore, any voter between $c_l^*$ and $c_r^*$ has a distance of at most $\opt$ to the closer candidate in $\{c_l^*,c_r^*\}$. As every voter is either between $c_l^*$ and $c_r^*$ or outside the interval $c_l',c_r'$, the distance of each voter to the closest candidate in our selected set is at most $\opt$ and we get a distortion of $1$.
\end{proof}

\section{Lower Bound on Distortion Factor}
\label{sec:LB}

\subsection{Lower bound for sum-cost}
\label{sec:sum-cost-LB}
In this section, we provide lower bounds for distortion in $k$-committee election when considering the sum-cost objective. 
First, we extend the lower bound of \cite{caragiannis2022metric} to show that achieving bounded distortion is not possible for general $k$ even if we allow to select $\omega(k)$ candidates. Furthermore, the lower bound holds for the line metric. 
Next, we show that for general metrics, it is not possible to achieve a distortion of $1$ when compared to the single optimal candidate if at least one candidate is not selected. We argue about this lower bound using instances in the two-dimensional plane (2-D Euclidean metric), showing that using the line metric space is crucial to our positive results.

\begin{theorem}
For $k\geq3$, there exists an instance for the $k$-committee election (with respect to the sum-cost) where any algorithm
choosing at most $\lceil\frac{k+1}{2}\rceil\lfloor\frac{k+1}{2}\rfloor-1$ candidates has unbounded distortion when compared to an optimal choice of $k$ candidates.
\end{theorem}
\begin{proof}
Let $a=\lceil\frac{k+1}{2}\rceil$ and $b=\lfloor\frac{k+1}{2}\rfloor$. We consider an instance of the problem with $ab$ candidates located on a line. We partition these candidates into $a$ groups $C_1,C_2,\ldots,C_a$ of size $b$ each. Let the $j$-th candidate in set $C_i$ be $c_{i,j}$. We consider a voter $v_{i,j}$ for each such candidate and set the preferences of voter $v_{i,j}$ as follows:
\begin{itemize}
    \item For any two candidates $c_{i',j'}$ and $c_{i'',j''}$ with $i'< i''$, $c_{i',j'}\succ c_{i'',j''}$ if $i''> i$ and $c_{i'',j''}\succ c_{i',j'}$ otherwise. This means that in voter $v_{i,j}$'s preference order, each set $C_i$ occupies a consecutive segment, with $C_i \succ C_{i-1} \succ \ldots \succ C_1 \succ C_{i+1}\succ\ldots\succ C_{a}$. (By the notation, $C_r \succ C_s$, we mean every candidate in the group $C_s$ is preferred over any candidate in the group $C_r$.)
    \item With $C_i$, the voter $v_{i,j}$'s preference order is 
    \[
    c_{i,j}\succ {c_{i,j-1}}\succ\ldots\succ c_{i,1}\succ c_{i,j+1} \succ \ldots\succ c_{i,b}.
    \]
    \item The preference of voter $v_{i,j}$ for candidates in $C_{i'}$ is in increasing order of $j$ if $i'>i$ and decreasing order of $j$ if $'<i$. 
\end{itemize}
Now, we show that any deterministic algorithm $\alg$ choosing at most $ab-1=\lceil\frac{k+1}{2}\rceil\lfloor\frac{k+1}{2}\rfloor-1$ candidates will have unbounded distortion when compared to an optimal choice of $a+b-1=k$ candidates. As $ab-1<ab$, at least one candidate $c_{i^*,j^*}$ is not chosen by $\alg$ 
in this instance.
Then, consider arbitrary constants $\ell > 2$ and $\epsilon>0$. Now, for any $i\neq i^*\in[a]$ and $j\in[b]$, we place candidate $c_{i,j}$ and voter $v_{i,j}$ in point $\ell 3^{i}+\epsilon 3^{j-b}$ if $i\neq i^*$ and point $\ell 3^{i}+ 3^{j-b}$ if $i= i^*$.

Next, we first show that this placing respects the voters' preference orders. We can see that for each $i \in [a],j\in[b] $, the distance of voter $v_{i,j}$ to any candidate in 
$C_{i+1}$ is at least 
\[
\ell3^{i+1}-\ell3^{i}-1\geq2\ell3^{i}  -1 \geq \ell3^{i} + \ell-1 > \ell3^{i}
\]
while the voter's distance to any candidate in $C_1$ is at most 
\[
\ell3^{i}+1-\ell3^{1} \leq l3^{i}.
\]
Therefore, any candidate in $C_1$ is closer to $v_{i,j}$ than any candidate in $C_{i+1}$. In addition, the distance of this voter to any candidate in $C_i$ is at most $1$, while its distance to the closest candidate in $C_{i-1}$ is at least 
\[
\ell3^i-\ell3^{i-1}-1 \geq 2\ell3^{i-1}-1\geq\ell>1.
\]
Now, since candidate groups $C_i'$ are ordered left to right on the line based on increasing $i'$, and voter $v_{i,j}$ is on the same point as $c_{i,j}$, we can say 
\[
C_i \succ C_{i-1} \succ \ldots \succ C_1 \succ C_{i+1}\succ\ldots\succ C_{a}
\]
holds for voter $v_{i,j}$. The preference for candidates within each $C_{i'}$ where $i'\neq i$ will also be consistent with the ordering, as they are all to the left or right of $v_{i,j}$. Finally, within $C_i$, the candidates $c_{i,j'}$ are ordered left to right in increasing order of $j'$, with voter $v_{i,j}$ coinciding with candidate $c_{i,j}$. Since the distance of $v_{i,j}$ to $v_{i,1}$ is less than its distance to $v_{i,j+1}$, the preference order for $v_{i,j}$ is as desired.

Now, we compare the cost of our algorithm and an optimal solution, choosing at most $a+b-1=k$ candidates. Since the algorithm does not choose $c_{i^*,j^*}$,
$v_{i^*,j^*}$ has a cost of at least $3^{j^*-b-1}\geq3^{-b}$ from its closest candidate.
So the cost of $\alg$ is at least 
$3^{-b}$. 
On the other hand, we can consider choosing every candidate in $C_{i^*}$, along with $c_{i,1}$ for every $i\neq i^*\in[a]$. For each voter $v_{i,j}$, if $i=i^*$, this choice has a cost of $0$ as $c_{i,j}$ is chosen, and if $i\neq i^*$, its distance to $c_{i,1}$ is at most $\epsilon3^{j-b}\leq\epsilon$. So, the total cost is at most $(a-1)b\epsilon$, and only $a+b-1=k$ candidates are selected. Now, the distortion will be at least 
\[
\frac{3^{-b}}{(a-1)b\epsilon}
\]
which can be arbitrarily increased to take any value by proper choice of $\epsilon$.  

\end{proof}

\begin{remark}
    We note that this instance has a recursive structure with two levels that can be generalized to more levels. In a generalization with $\ell$ levels, we can define each voter and candidate using a sequence of $\ell$ indices with values in $[a]$. Then, a voter $v_{i_1,\ldots,i_\ell}$'s preference order for candidates $c_{j_1,\ldots,j_\ell}$ is determined by the first index where sequences $i_1,\ldots,i_\ell$ and $j_1,\ldots,j_\ell$ defer. If this index is different for two candidates, the one with more matching indices is preferred, and if both differ with the voter at the same index $i$, the same pattern of preference $i\succ i-1\succ\ldots\succ1\succ i+1\succ \ldots\succ a$ is used. Then, suppose $c_{i_1^*,\ldots,i_\ell^*}$ is a candidate not chosen 
    by an algorithm selecting at most $a^\ell-1$ candidates. In that case, we can place the candidates and voters with a similar recursive structure such that the set of candidates $\{c_{i_1^*,\ldots,i_{r-1}^*,i_{r},1,\ldots,1}\mid r \in [\ell],i_{r} \in [a] \setminus \{i_{r}^*\}\}\cup \{c_{i_1^*,\ldots,i_\ell^*}\}$ has as a multiple of $\epsilon$ as its total cost. In contrast, the closest candidate except $c_{i_1^*,\ldots,i_\ell^*}$ to $v_{i_1^*,\ldots,i_\ell^*}$ has a constant distance to it not dependent on $\epsilon$. This leads to an unbounded distortion when comparing any algorithm selecting at most $a^\ell-1$ candidates to the optimal selection of $\ell(a-1)+1$ candidates. In particular, choosing $a=2$ and $\ell=\lceil \log_2{k+1}\rceil$, we can show that any algorithm for selecting $k$ candidates cannot achieve an unbounded distortion when compared to the optimal selection of $\lceil \log_2(k+1)\rceil +1$ candidates.
\end{remark}

Next, We construct an example in a two-dimensional plane, where picking even $m-1$ candidates out of $m$ candidates won't guarantee that an optimal candidate is chosen. Furthermore, we provide construction of such an instance so that if we consider the distance to the closest candidate among the $m-1$ selected candidates for each voter, we still get a distortion of at least $1 + \frac{2}{m-1} - \epsilon$, for any given $\epsilon > 0$, compared to the optimal single candidate.

\begin{theorem}
    \label{thm:LB-sum-all}
    Any deterministic algorithm for the $(m-1)$-committee election problem (with respect to the sum-cost) that, given voters' preference orderings on $m$ candidates, selects at most $m-1$ candidates must have a 1-distortion factor of at least $ 1 + \frac{2}{m-1} - \epsilon$, for any $\epsilon > 0$.
\end{theorem}
\begin{proof}
    We consider a family of $m+1$ instances $I_0,I_1,\cdots,I_m$ with $m$ voters $\left\{v_1,v_2,\cdots,v_m\right\}$ and $m$ candidates $\left\{c_1,c_2,\cdots,c_m\right\}$ on a two-dimensional plane. In terms of voters' and candidates' locations on the plane, each of these instances is distinct. However, in terms of voters' preference orders on candidates (i.e., ordinal information), $I_0,I_1,\cdots,I_m$ are the same. To be more specific, let us describe these instances.

    Let $\ell=3/\epsilon - 1$. We first describe the instance $I_0$. For each $i \in [m]$, the candidate $c_i$ is located at the point $\left(-\ell,2^{i-1}/2^m\right)$. For each $i \in [m]$, the voter $v_i$ is located at the point $\left(0,2^{i-1}/2^m\right)$. (See Figure~\ref{fig:lower}.) It is straightforward to observe that, for each $i \in [m]$, the preference order of the voter $v_i$ is 
    \[
    c_i \succ c_{i-1}\succ \cdots\succ c_1 \succ c_{i+1} \succ c_{i+2}\succ \cdots \succ c_m. 
    \]

    Next, for each $j \in [m]$, we define the instance $I_j$ as follows: For each $i \ne j \in [m]$, the locations of the candidate $c_i$ and the voter $v_i$ are the same as that in the instance $I_0$. Both the candidate $c_j$ and the voter $v_j$ are located at the point $\left(\ell,2^{i-1}/2^m \right)$. (See Figure~\ref{fig:lower}.)  It is easy to observe that for any voter $v_i$ (where $i \ne j$), the distance to any candidate $c_j$ remains the same (as in $I_0$). Since the locations of other candidates remain the same, the preference order for any voter $v_i$ (where $i \ne j$) remains the same as in $I_0$. Let us now consider the voter $v_j$. Note that its closest candidate is $c_j$, and the distance to any other candidate $c_i$ remains the same as in $I_0$. Thus, the preference order of the voter $v_j$ is 
    \[
    c_j \succ c_{j-1}\succ \cdots\succ c_1 \succ c_{j+1} \succ c_{j+2}\succ \cdots \succ c_m 
    \]
    which is the same as that in $I_0$.

    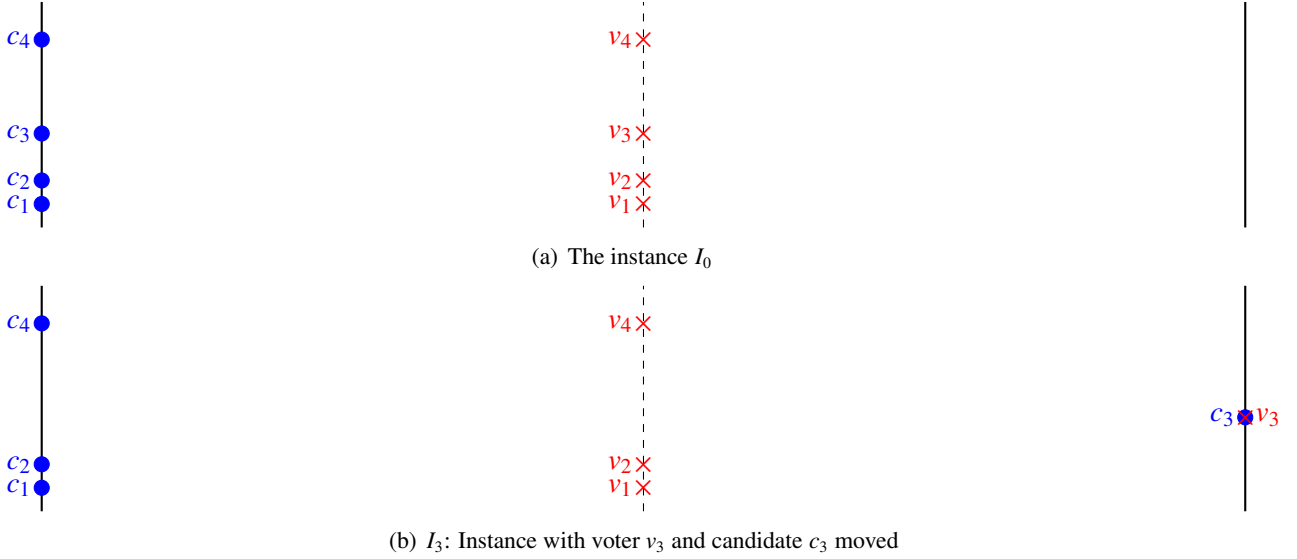
\begin{figure}[t]
        \centering
\subfigure[The instance $I_0$]{
\begin{tikzpicture}
    \def\m{5} 
    \def\d{8} 

    \draw[thick] (-\d, 0) -- (-\d, 3);
    \draw[thick] (\d, 0) -- (\d, 3);
    
    \draw[dashed] (0, 0) -- (0, 3);
    
    \foreach \i in {1,2,3,4} {
        \draw[blue, fill=blue] (-\d, {5 * 2^\i/2^\m}) circle (0.1) node[left] {$c_\i$};
    }

    \foreach \i in {1,2,3,4} {
        \draw[red, thick] (0, {5 * 2^\i/2^\m}) node[left] {$v_\i$};
        \draw[red, thick] (0, {5 * 2^\i/2^\m}) node {$\times$};
    }
\end{tikzpicture}
}
\subfigure[$I_3$: Instance with voter $v_3$ and candidate $c_3$ moved]{
\begin{tikzpicture}
    \def\m{5} 
    \def\d{8} 

    \draw[thick] (-\d, 0) -- (-\d, 3);
    \draw[thick] (\d, 0) -- (\d, 3);
    
    \draw[dashed] (0, 0) -- (0, 3);
    
    \foreach \i in {1,2,4} {
        \draw[blue, fill=blue] (-\d, {5 * 2^\i/2^\m}) circle (0.1) node[left] {$c_\i$};
    }
    \draw[blue, fill=blue] (\d, {5 * 2^3/2^\m}) circle (0.1) node[left] {$c_3$};
    \foreach \i in {1,2,4} {
        \draw[red, thick] (0, {5 * 2^\i/2^\m}) node[left] {$v_\i$};
        \draw[red, thick] (0, {5 * 2^\i/2^\m}) node {$\times$};
    }
    \draw[red, thick] (\d, {5 * 2^3/2^\m}) node[right] {$v_3$};
    \draw[red, thick] (\d, {5 * 2^3/2^\m}) node {$\times$};
\end{tikzpicture}
}
        \caption{Figures of the lower bound instances. In (a), all candidates are located on the line $x=-\ell$, with voters with matching $y$-coordinates on the line $x=0$. In (b), voter $v_3$ and candidate $c_3$ are moved to the line $x=\ell$, while keeping the same $y$-coordinate.}
        \label{fig:lower}
    \end{figure}

    Let us now consider any arbitrary deterministic algorithm $\alg$ that selects at most $m-1$ candidates. Now, suppose given the preference orders of voters as in $I_0$ (the same as in $I_1,\cdots,I_m$), $\alg$ selects a set $C$ of candidates where $|C| \le m-1$. Suppose $c_k \not \in C$, for some $k \in [m]$. Now, consider the instance $I_k$. Note, since voters' preference orders, i.e., ordinal information, are the same as in $I_0$, $\alg$ also selects the set $C$ for the instance $I_k$. Then 
    \begin{align*}
        \cost_s(C,I_k) & = \sum_{i \in [m]} d(v_i,C)\\
        & \ge \sum_{i \ne k} \ell + 2\ell = (m+1) \ell. 
    \end{align*}

    On the other hand, selecting only the candidate $c_k$ for the instance $I_k$ would lead to the cost of
    \[
    \cost_s(c_k,I_k) = \sum_{i \in [m]} \left|\left| v_i - c_k  \right|\right|_2 \le \sum_{i \ne k} (\ell+1) = (m-1)\left(\ell + 1\right)
    \]
    and thus the optimal cost $\opt(I_k) \le (m-1)\left(\ell + 1\right)$.

    Hence, the distortion factor of $\alg$ on the instance $I_k$ is 
    \begin{align*}
    \dist(\alg)&= \frac{\cost_s(C,I_k)}{\opt(I_k)} \\
    &\geq \frac{(m+1)\ell}{(m-1)(\ell+1)}\\ 
    &=\left(1 + \frac{2}{m-1}\right)\left(1 - \frac{1}{\ell+1}\right)\\
    &=1 + \frac{2}{m-1} - \frac{1}{\ell+1} \left(1 + \frac{2}{m-1}\right)\\
    & \ge 1 + \frac{2}{m-1} - \epsilon &&\text{(since }\ell=3/\epsilon - 1\text{)}.
\end{align*}
\end{proof}

Next, we generalize our lower bound to any algorithm that picks at most $r$ candidates.

\subsection{Lower bounds for the max-cost}
\label{sec:max-cost-LB}
In this section, we provide lower bounds on the distortion for the committee election problem when considering the max-objective. In the 2-D Euclidean space, we show that even when choosing $m-1$ candidates out of $m$, we cannot guarantee a distortion of less than $3$, matching the distortion achieved by simple algorithms when selecting only one candidate. In addition, we provide lower bounds when considering the line metric, showing that our algorithms that select at most $k=2$ and $k=3$ candidates, respectively, achieve tight bounds in terms of distortion. 

We first show that any deterministic algorithm choosing at most $k < m$ candidates (out of $m$ candidates) cannot achieve a distortion strictly lower than $3$ when compared to a single optimal candidate, even when the candidates and voters are located in the two-dimensional plane. 
\begin{theorem}
    \label{thm:max-cost-lb}
    Any deterministic algorithm for the $k$-committee election (with respect to the max-cost) that, selects at most $k < m$ candidates out of $m$ candidates, must have a 1-distortion of at least $3-\epsilon$ for any $\epsilon>0$. 
\end{theorem}
\begin{proof}
    We proceed similarly to Theorem \ref{thm:LB-sum-all}, using a family of instances $I_0,I_1,\ldots,I_m$ on the two-dimensional plane, such that the voters' ordinal preferences are the same in all of these instances, but the locations and distances of candidates and voters vary. In each instance, we have $m$ voters $\{v_1,v_2,\ldots,v_m\}$ and $m$ candidates $\{c_1,c_2,\ldots,c_m\}$.

    Let $\ell=3/\epsilon-1$. We use the same instance $I_0$ as in Theorem \ref{thm:LB-sum-all}, where for each $i\in[m]$, candidate $c_i$ is located at $(-\ell,2^{i-1}/2^m)$ and voter $v_i$ is located at $(0,2^{i-1}/2^m)$. Then, for each $i \in [m]$, voter $v_i$ will have the ordinal preference
    \[
    c_i \succ c_{i-1}\succ \cdots\succ c_1 \succ c_{i+1} \succ c_{i+2}\succ \cdots \succ c_m. 
    \]
    in instance $I_0$.

    Next, we define instance $I_j$ for each $j\in[m]$. For each $i\in [m] \setminus \{j\}$, $v_i$ and $c_i$ will remain in the same location as in $I_0$. Meanwhile, candidate $c_j$ is moved to $(\ell,2^{j-1}/2^m)$ and voter $v_j$ to $(2\ell,2^{j-1}/2^m)$. It can be seen that for any $i\neq j$, the distance of voter $v_i$ to candidate $c_j$ will be unchanged compared to $I_0$, and therefore the voter's ordinal preference will remain the same. For voter $j$, its closest candidate will remain $c_j$. In addition, its preference for the other candidates will remain unchanged. Therefore, voter $v_j$'s ordinal preference will remain the same as in $I_0$ too and all voters' ordinal preferences are identical in $I_0$ and $I_j$. An example of these instances is illustrated in Figure \ref{fig:lower-max}.

    Now, consider an arbitrary deterministic algorithm $\alg$ that selects $k<m$ candidates. Since $k<m$, when running $\alg$ on $I_0$, there exists a candidate $c_j$ that is not in the set of selected candidates $C$. Now, consider the algorithm's performance on $I_j$. Since the voter's ordinal preferences are the same in $I_0$ and $I_j$, and $\alg$ operates using only this information, its output on $I_j$ must be the same as in $I_0$ and therefore it does not select $c_j$. Now, we can lower bound the cost of $\alg$ with respect to the max objective as
    \begin{align*}
        \cost_m(C,I_0)
        &\geq d(v_j, C)\\
        &\geq 3\ell. \tag{Since $c_j\not\in C$}
    \end{align*}
    On the other hand, if we only select $c_j$, we get
    \begin{align*}
        \cost_m(c_j,I_0)&=\max_{i\in[m]}\left|\left| v_i - c_j  \right|\right|_2 \\
        &\leq \ell + 1.
    \end{align*}
    This shows that the optimal cost $\opt_m(I_j) \leq \ell + 1$. Finally, combining these two, we get that the distortion of $\alg$ on instance $I_j$ is at least
    \begin{align*}
        \dist(\alg)&\geq \frac{\cost_m(C,I_k)}{\opt_m(I_k)}\\
        &\geq \frac{3\ell}{\ell+1}\\
        &=3-\frac{3}{\ell+1}\\
        &=3-\epsilon. \tag{$\ell=3/\epsilon-1$}
    \end{align*}
\end{proof}

\begin{figure}[t]
        \centering
\subfigure[The instance $I_0$]{
\begin{tikzpicture}
    \def\m{5} 
    \def\d{5} 

    \draw[thick] (-\d, 0) -- (-\d, 3);
    \draw[thick] (\d, 0) -- (\d, 3);
    
    \draw[dashed] (0, 0) -- (0, 3);
    \draw[dashed] (2 * \d, 0) -- (2 * \d, 3);

    \foreach \i in {1,2,3,4} {
        \draw[blue, fill=blue] (-\d, {5 * 2^\i/2^\m}) circle (0.1) node[left] {$c_\i$};
    }

    \foreach \i in {1,2,3,4} {
        \draw[red, thick] (0, {5 * 2^\i/2^\m}) node[left] {$v_\i$};
        \draw[red, thick] (0, {5 * 2^\i/2^\m}) node {$\times$};
    }
\end{tikzpicture}
}
\subfigure[$I_3$: Instance with voter $v_3$ and candidate $c_3$ moved]{
\begin{tikzpicture}
    \def\m{5} 
    \def\d{5} 

    \draw[thick] (-\d, 0) -- (-\d, 3);
    \draw[thick] (\d, 0) -- (\d, 3);
    
    \draw[dashed] (0, 0) -- (0, 3);
    \draw[dashed] (2 * \d, 0) -- (2 * \d, 3);
    \foreach \i in {1,2,4} {
        \draw[blue, fill=blue] (-\d, {5 * 2^\i/2^\m}) circle (0.1) node[left] {$c_\i$};
    }
    \draw[blue, fill=blue] (\d, {5 * 2^3/2^\m}) circle (0.1) node[left] {$c_3$};
    \foreach \i in {1,2,4} {
        \draw[red, thick] (0, {5 * 2^\i/2^\m}) node[left] {$v_\i$};
        \draw[red, thick] (0, {5 * 2^\i/2^\m}) node {$\times$};
    }
    \draw[red, thick] (2 * \d, {5 * 2^3/2^\m}) node[right] {$v_3$};
    \draw[red, thick] (2 * \d, {5 * 2^3/2^\m}) node {$\times$};
\end{tikzpicture}
}
        \caption{Figures of the lower bound instances for the max objective. In (a), all candidates are located on the line $x=-\ell$, with voters with matching $y$-coordinates on the line $x=0$. In (b), candidate $c_3$ is moved to the line $x=\ell$, and voter $v_3$ is moved to the line $x=2\ell$ while keeping the same $y$-coordinates.}
        \label{fig:lower-max}
    \end{figure}
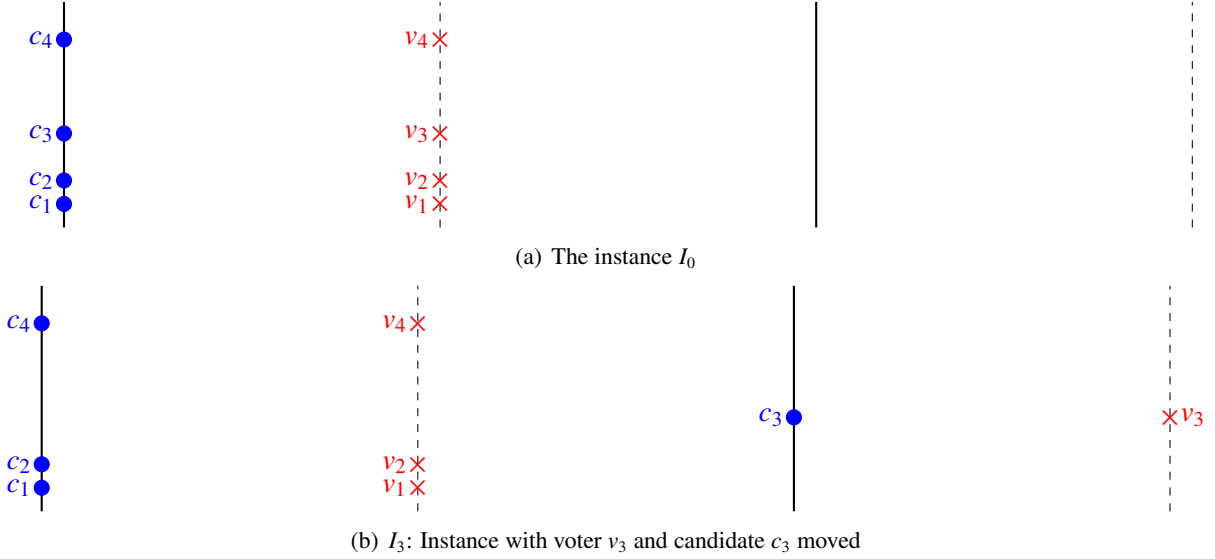

Next, we show that our upper bounds for the line metric are tight by proving matching lower bounds. We begin with the case $k=2$.

\begin{theorem}
    \label{thm:lb-max-line2}
    Any deterministic algorithm for the 2-committee election (with respect to the max-cost) when voters and candidates are located on a line must have a 1-distortion of at least $2-\epsilon$ for any $\epsilon>0$. 
\end{theorem}
\begin{proof}
    
We consider the following instance with three voters $\{v_1,v_2,v_3\}$ and three candidates $\{c_1,c_2,c_3\}$. 
Let the preferences of voters be as follows:
\begin{align*}
    v_1: c_1\succ c_2\succ c_3\\
    v_2: c_2\succ c_1\succ c_3\\
    v_3: c_3\succ c_2\succ c_1.
\end{align*}
Next, we consider three possible placements for the voters and candidates that respect the above (desired) preference order of voters, as shown in Figure \ref{fig:lb-max-line2}. In all instances, $v_1$ is located at point $-1$, and $v_3$ at point $1$. In the first instance, $v_2$ is located at $0.5$. Candidates $c_1,c_2$, $c_3$ are located at points $0,1-\epsilon$ and $1+\epsilon/2$ respectively. It is easy to see that the voters' preferences in this instance match our desired orders.
Now, we can see that in this instance, $c_1$ has a distance of at most $1$ to every voter, while both $c_2$ and $c_3$ have a distance of at least $2-\epsilon$ from $v_1$. Therefore, not choosing $c_1$ will lead to a distortion of at least $2-\epsilon$.
\begin{figure}[t]
        \centering
\subfigure[First instance: Choosing $c_1$ achieves a cost of $1$, while the other two candidates have a distances of at least $2-\epsilon$ to $v_1$.]{
\begin{tikzpicture}
    \def\d{3} 
    \def\l{8} 

    \draw[thick] (-\l, 0) -- (\l, 0);
    
    \draw[blue, fill=blue] (0, 0) circle (0.1) node[below] {$c_1$};
    \draw[blue, fill=blue] (\d-0.4, 0) circle (0.1) node[below] {$c_2$};
    \draw[blue, fill=blue] (\d + 0.3, 0) circle (0.1) node[below] {$c_3$};
    \foreach \i in {1,3} {
        \draw[red, thick] (-2 * \d +\i * \d, 0) node[above] {$v_\i$};
        \draw[red, thick] (-2 * \d +\i * \d, 0) node {$\times$};
    }
    \draw[red, thick] (\d *0.5, 0) node[above] {$v_2$};
        \draw[red, thick] (\d * 0.5, 0) node {$\times$};
\end{tikzpicture}
}
\subfigure[Second instance: Choosing $c_2$ achieves a cost of $1$, while the other two candidates have a distance of at least $2-\epsilon$ to $v_2$.]{
\begin{tikzpicture}
    \def\d{3} 
    \def\l{8} 

    \draw[thick] (-\l, 0) -- (\l, 0);
    
    \draw[blue, fill=blue] (-2*\d + 0.1, 0) circle (0.1) node[below] {$c_1$};
    \draw[blue, fill=blue] (0, 0) circle (0.1) node[below] {$c_2$};
    \draw[blue, fill=blue] (2*\d - 0.05, 0) circle (0.1) node[below] {$c_3$};
    \foreach \i in {1,2,3} {
        \draw[red, thick] (-2 * \d +\i * \d, 0) node[above] {$v_\i$};
        \draw[red, thick] (-2 * \d +\i * \d, 0) node {$\times$};
    }
\end{tikzpicture}
}
\subfigure[Third instance: Choosing $c_3$ achieves a cost of $1$, while the other two candidates have a distances of at least $2-\epsilon$ to $v_3$.]{
\begin{tikzpicture}
    \def\d{3} 
    \def\l{8} 

    \draw[thick] (-\l, 0) -- (\l, 0);
    
    \draw[blue, fill=blue] (-\d, 0) circle (0.1) node[below] {$c_1$};
    \draw[blue, fill=blue] (-\d + 0.3, 0) circle (0.1) node[below] {$c_2$};
    \draw[blue, fill=blue] (0, 0) circle (0.1) node[below] {$c_3$};
    \foreach \i in {1,3} {
        \draw[red, thick] (-2 * \d +\i * \d, 0) node[above] {$v_\i$};
        \draw[red, thick] (-2 * \d +\i * \d, 0) node {$\times$};
    }
    \draw[red, thick] (-\d + 0.3, 0) node[above] {$v_2$};
        \draw[red, thick] (-\d + 0.3, 0) node {$\times$};
\end{tikzpicture}
}
        \caption{Possible locations of the voters and candidates in the lower bound instance for $k=2$. }
        \label{fig:lb-max-line2}
    \end{figure}
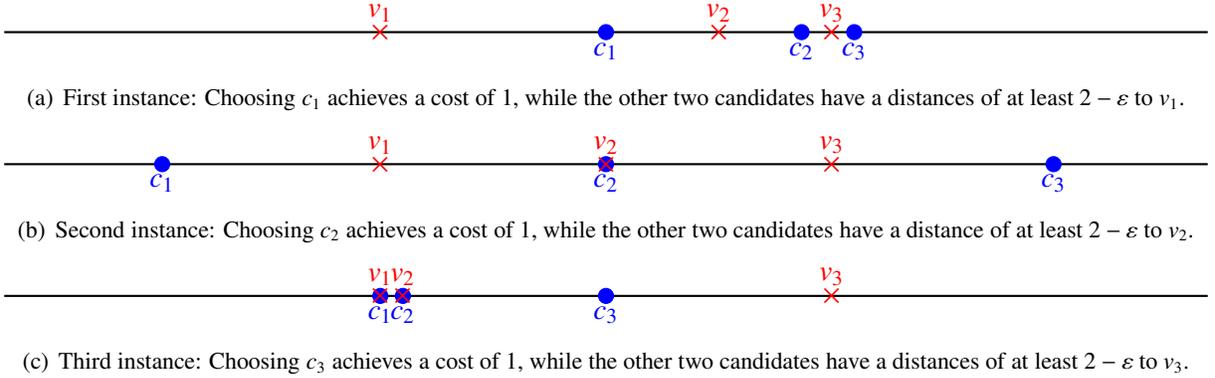

In the second instance, we place $v_2$ and $c_2$ at point $0$, $c_1$ at point $-2+\epsilon$, and $c_3$ at point $2-\epsilon/2$. Again, we can see that the voters' preferences will follow the desired orders. In this instance, $c_2$ has a distance of at most $1$ from all voters, while $v_2$ is at a distance of at least $2-\epsilon$ from the other candidates. Therefore, not choosing $c_2$ will lead to a distortion of at least $2-\epsilon$.

Finally, in the third instance, we place $c_1$ at point $-1$, $c_2$ and $v_2$ at point $-1+\epsilon$ and $c_3$ at point $0$. This placement will also respect the desired preference orders for each voter. Additionally, since $c_3$ has a distance of at most $1$ from all voters and $v_3$ has a distance of at least $2-\epsilon$ from the other candidates, not choosing $c_3$ leads to a distortion of at least $2-\epsilon$.

Now, since these instances cannot be distinguished based on the voters' ordinal preferences, and not choosing any of the candidates leads to a distortion of at least $2-\epsilon$, any deterministic algorithm selecting two candidates cannot achieve a distortion better than $2-\epsilon$.

\end{proof}


Lastly, we provide a tight lower bound for the case $k=3$.

\begin{theorem}
    \label{thm:lb-max-line3}
    Any deterministic algorithm for the 3-committee election (with respect to the max-cost) when voters and candidates are located on a line must have a 1-distortion of at least $3/2-\epsilon$ for any $\epsilon>0$. 
\end{theorem}
\begin{proof}
    
We consider instances $\{I_1,I_2,I_3,I_4\}$ each with four voters $\{v_1,v_2,v_3,v_4\}$ and four candidates $\{c_1,c_2,c_3,c_4\}$ such that the preferences of voters are as follows in every instance:
\begin{align*}
    v_1: c_1\succ c_2\succ c_3\succ c_4\\
    v_2: c_2\succ c_1\succ c_3 \succ c_4\\
    v_3: c_3\succ c_4\succ c_2\succ c_1\\
    v_4: c_4\succ c_3\succ c_2\succ c_1.
\end{align*}
For instance $I_i$, we choose the location of voters and candidates so that candidate $c_i$ has a distance of at most $2$ to each voter, while voter $v_i$ has a distance of at least $3-2\epsilon$ to every candidate except $c_i$. This leads to a distortion of at least $3/2 - \epsilon$ for any deterministic algorithm $\alg$, as these instances cannot be distinguished based on ordinal preferences, and at least one candidate $c_i$ is not chosen in instance $I_i$ by $\alg$. These instances are shown in Figure \ref{fig:lb-max-line3}.

In instance $I_1$, we have voters $v_1,v_2,v_3,v_4$ located at points $-2,1-\epsilon,2-\epsilon,2$ and candidates $c_1,c_2,c_3,c_4$ located at points $0,1,2-\epsilon,2$ respectively. It can be seen that the ordinal preference of each voter in this instance matches the desired ordering. Now, $c_1$ has a distance of at most $2$ to every voter in this instance, while $v_1$ is at a distance of at least $3$ from every candidate except $c_1$. So, if $c_1$ is not chosen, we get a distortion of at least $3/2$.

In instance $I_2$, we have voters $v_1,v_2,v_3,v_4$ located at points $-2,-1-\epsilon,2-\epsilon,2$ and candidates $c_1,c_2,c_3,c_4$ located at points $-4+\epsilon,0,2-\epsilon,2$ respectively. Once again, we can see that the ordering for voters' preferences is respected:
\begin{align*}
    &d(v_1,c_1)=2-\epsilon < d(v_1,c_2)=2 < d(v_1,c_3)=4-\epsilon < d(v_1,c_4)=4\\
    &d(v_2,c_2)=1+\epsilon < d(v_2,c_1)=3-2\epsilon < d(v_2,c_3)=3 < d(v_1,c_4)=3+\epsilon\\
    &d(v_3,c_3)=0 < d(v_3,c_4)=\epsilon < d(v_3,c_2)=2-\epsilon < d(v_3,c_1)=6-2\epsilon\\
    &d(v_4,c_4)=0 < d(v_4,c_3)=\epsilon < d(v_4,c_2)=2 < d(v_4,c_1)=6-\epsilon.
\end{align*}
In addition, since $c_2$ is at a distance of at most $2$ from every voter and every candidate except $c_2$ has a distance of at least $3-2\epsilon$ from $v_2$, not choosing $c_2$ leads to a distortion of at least $3/2-\epsilon$.

Based on the symmetry in voters' preferences between $c_1,c_2$ and $c_4,c_3$, we create mirror versions of $I_1$ and $I_2$ as $I_4$ and $I_3$ respectively, so that not choosing $c_4$ or $c_3$ would lead to a distortion of at least $3/2 - \epsilon$. Therefore, since any deterministic algorithm must omit one of the candidates, we cannot achieve a distortion of better than $3/2-\epsilon$ in this case.
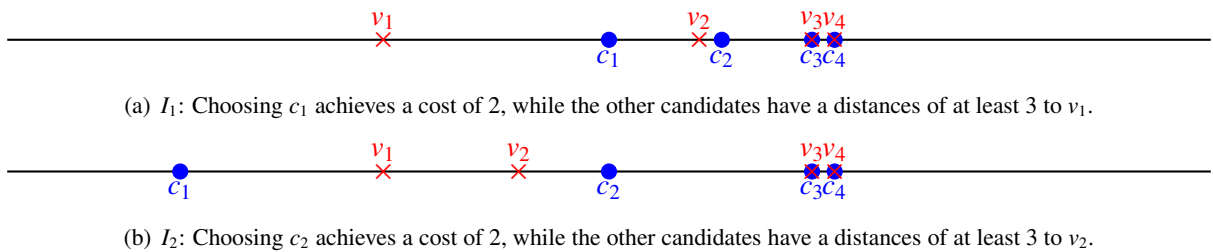
\begin{figure}[t]
        \centering
\subfigure[$I_1$: Choosing $c_1$ achieves a cost of $2$, while the other candidates have a distances of at least $3$ to $v_1$.]{
\begin{tikzpicture}
    \def\d{1.5} 
    \def\l{8} 

    \draw[thick] (-\l, 0) -- (\l, 0);
    
    \draw[blue, fill=blue] (0, 0) circle (0.1) node[below] {$c_1$};
    \draw[blue, fill=blue] (\d, 0) circle (0.1) node[below] {$c_2$};
    \draw[blue, fill=blue] (2 * \d - 0.3, 0) circle (0.1) node[below] {$c_3$};
    \draw[blue, fill=blue] (2 * \d, 0) circle (0.1) node[below] {$c_4$};
    \draw[red, thick] (-2 * \d , 0) node[above] {$v_1$};
    \draw[red, thick] (-2 * \d, 0) node {$\times$};
    \draw[red, thick] (\d-0.3, 0) node[above] {$v_2$};
    \draw[red, thick] (\d-0.3, 0) node {$\times$};
    \draw[red, thick] (2 * \d -0.3, 0) node[above] {$v_3$};
    \draw[red, thick] (2 * \d -0.3, 0) node {$\times$};
    \draw[red, thick] (\d *2, 0) node[above] {$v_4$};
    \draw[red, thick] (\d * 2, 0) node {$\times$};
\end{tikzpicture}
}
\subfigure[$I_2$: Choosing $c_2$ achieves a cost of $2$, while the other candidates have a distances of at least $3$ to $v_2$.]{
\begin{tikzpicture}
    \def\d{1.5} 
    \def\l{8} 

    \draw[thick] (-\l, 0) -- (\l, 0);
    
    \draw[blue, fill=blue] (-4 * \d + 0.3, 0) circle (0.1) node[below] {$c_1$};
    \draw[blue, fill=blue] (0, 0) circle (0.1) node[below] {$c_2$};
    \draw[blue, fill=blue] (2 * \d - 0.3, 0) circle (0.1) node[below] {$c_3$};
    \draw[blue, fill=blue] (2 * \d, 0) circle (0.1) node[below] {$c_4$};
    \draw[red, thick] (-2 * \d , 0) node[above] {$v_1$};
    \draw[red, thick] (-2 * \d, 0) node {$\times$};
    \draw[red, thick] (-\d+0.3, 0) node[above] {$v_2$};
    \draw[red, thick] (-\d+0.3, 0) node {$\times$};
    \draw[red, thick] (2 * \d -0.3, 0) node[above] {$v_3$};
    \draw[red, thick] (2 * \d -0.3, 0) node {$\times$};
    \draw[red, thick] (\d *2, 0) node[above] {$v_4$};
    \draw[red, thick] (\d * 2, 0) node {$\times$};
\end{tikzpicture}
}        \caption{Possible locations of the voters and candidates in the lower bound instance for $k=3$. }
        \label{fig:lb-max-line3}
    \end{figure}
\end{proof}

\bibliographystyle{abbrv}
\bibliography{resources}
\end{document}

%% file: macros.tex
\usepackage{amsthm}
\usepackage{framed}
\usepackage{mdframed}

\theoremstyle{definition}

\colorlet{shadecolor}{gray!20}

\algnewcommand{\algorithmicand}{\textbf{ and }}
\algnewcommand{\algorithmicor}{\textbf{ or }}
\algnewcommand{\OR}{\algorithmicor}
\algnewcommand{\AND}{\algorithmicand}
\algdef{SE}[DOWHILE]{Do}{doWhile}{\algorithmicdo}[1]{\algorithmicwhile\ #1}

\newtheorem{theorem}{Theorem}
\newtheorem{corollary}[theorem]{Corollary}
\newtheorem{lemma}[theorem]{Lemma}

\newtheorem{definition}[theorem]{Definition}
\newtheorem{remark}[theorem]{Remark}

\newcommand{\GE}[1]{\succ_{#1}}
\newcommand{\last}{\textit{last}}
\newcommand{\first}{\textit{first}}
\newcommand{\order}{\textit{order}}
\newcommand{\sorted}{\textit{sorted}}

\newcommand{\elec}{\mathcal{E}}

\newcommand{\determine}{\textsc{Determine}}

\newcommand{\splitline}{\textsc{SplitLine}}
\newcommand{\sortcandidates}{\textsc{SortCandidates}}
\newcommand{\sortvoters}{\textsc{SortVoters}}
\newcommand{\sortcandiatesandvoters}{\textsc{SortCandiatesAndVoters}}

\newcommand{\Det}{C^*}
\newcommand{\New}{C_{new}}
\newcommand{\prefix}{P}

\newcommand{\core}{core}
\newcommand{\Core}{Core}

\renewcommand{\epsilon}{\varepsilon}
\newcommand{\alg}{\texttt{ALG}}
\newcommand{\opt}{\texttt{OPT}}
\newcommand{\copt}{\ensuremath{c_{\text{opt}}}}
\newcommand{\cost}{\texttt{cost}}
\newcommand{\dist}{\texttt{1-distortion}}

\newcommand{\redcross}{{\color{red}\text{\sffamily X}}}
\newcommand{\greencheck}{{\color{Green}\checkmark}}

%% file: intro.tex
\section{Introduction}

One of the fundamental challenges in the social choice theory is to elect representatives based on voters' preferences, ideally represented by cardinal utility functions that assign numerical values to each outcome. However, in most real-world scenarios, voters only provide ordinal information, such as preference orders among outcomes/candidates. This raises a natural question of how worse, if at all, a voting mechanism performs given ordinal information than cardinal information. Procaccia and Rosenschein~\cite{procaccia2006distortion} introduced the notion of \emph{distortion} to measure such an efficiency loss --  how different voting rules respond to the lack of cardinal information. Many practical voting scenarios can be formulated by considering both voters and candidates lying on a metric space (see~\cite{enelow1984spatial}). The distance to candidate locations determines voters' cardinal preferences for candidates -- voters rank candidates based on ascending distance, with the closest candidate being the most preferable and the farthest candidate being the least preferable. The worst-case behavior of any ordinal preference order-based voting rule/mechanism is captured by the notion of \emph{metric distortion}, introduced by Anshelevich \emph{et al.}~\cite{anshelevich2018approximating}. A voting mechanism, without access to the actual distances among the set of voters and candidates, seeks to minimize a specific cost function, which depends on the distances. Distortion is defined in relation to this cost: For any voting mechanism $f$, its distortion is the worst-case ratio (across all instances) of the cost of the solution produced by $f$ compared to the optimal cost.

Given a single fixed candidate, the cost for a voter is defined as its distance from the given candidate. Then, the overall cost is set to be an objective that combines these values across all the voters. Depending on the specific context in the literature, the following two objectives have widely been considered: A utilitarian objective, which aims to minimize the total individual costs for all voters, and an egalitarian cost, which minimizes the maximum cost experienced by any voter. Different variants of the metric distortion problem under the above objective functions have received significant attention, e.g.~\cite{goel2017metric, kempe2020analysis, gkatzelis2020resolving, anagnostides2022dimensionality, kizilkaya2023generalized}.

For the classical metric distortion problem, it is known that a distortion of $3$ can be achieved for any metric. Further, it is widely known that the candidate chosen by any deterministic method cannot achieve a distortion factor of less than 3, even in a line metric.
In other words, without knowing the exact distances, it is not possible to obtain a better than $3$-approximation to the optimal cost. This leads to an intriguing question: Can we obtain a better approximation (distortion) by selecting more than one candidate? Specifically, assume that the algorithm is allowed to choose $k > 1$ candidates, and we set the cost for each voter to be its distance to the closest chosen candidate. Can we design an algorithm for which the overall cost (either utilitarian or egalitarian) is at most $\alpha$ times the cost of an optimal candidate for some $\alpha < 3$? 

In this paper, we answer the above question affirmatively for the line metric. We not only attain better distortion by allowing the selection of more than one candidate but, in fact, attain optimal cost with $O(1)$ candidates.
We complement our upper-bound results with various lower-bound constructions showing impossibility results in line and 2-D Euclidean metric.

Our work provides a \emph{bicriteria} perspective to metric distortion, which, to our knowledge, has not been considered before. 
The pursuit of improved \emph{bi-criteria} approximation results for various classical optimization problems has already been well-investigated in the literature. Indeed,
when the metric is known, and the goal is to pick $k$ candidates that minimize the overall cost across voters, the election problem becomes an instance of either the $k$-median or the $k$-center clustering (for the utilitarian and egalitarian objective, respectively).
For these problems, numerous (constant-factor) approximation algorithms are known which select up to $O(k)$ centers (instead of $k$ centers), e.g.~\cite{feldman2007bi, wei2016constant, alamdari2018bicriteria}. Therefore, it is quite natural to explore a similar question in the metric distortion problem, where the underlying metric is not directly given. 

Our work additionally extends the existing line of research on $k$-committee elections~\cite{faliszewski2017multiwinner, elkind2017properties, caragiannis2022metric}, which also select a committee of $k$ candidates and aim to minimize some loss function across all voters. The key distinction is that we consider a single candidate as our baseline for calculating distortion, while these works consider a baseline of $k$ candidates (see Section~\ref{sec:related_work}). Specifically, for the utilitarian cost we consider (i.e., the sum of distances of voters to their closest candidate in the committee), Caragiannis, Shah, and Voudouris~\cite{caragiannis2022metric} show that the ratio of the cost for any deterministic method compared to the cost of an optimal committee can be unbounded in the worst case. This result fails to suggest a choice of candidates, as all possible choices are the same in the worst case. In contrast, we show that by choosing the right baseline, one can make a meaningful distinction between different choices even though the underlying metric is unknown. 

\subsection{Our contribution}

Our first result shows that when the underlying metric is a line, a committee of two candidates can achieve a 1-distortion of $1$ under the \emph{sum-cost} objective. Here, \emph{1-distortion} refers to the worst-case ratio between the cost of the selected committee and the cost of an optimal single-winner candidate.  
In fact, we will prove a stronger result by showing that it is possible to output a list of $2$ candidates that always contains an optimal one.
Formally, we prove the following theorem.
\begin{theorem}
    There exists an algorithm for the 2-committee utilitarian election on the line metric that is guaranteed to choose an optimum candidate in the elected committee. 
    Consequently, the 1-distortion of the algorithm is $1$. 
\end{theorem}

We further show that when the metric is not a line, one can obtain a distortion of $1+\frac{2}{m-1}$. 
\begin{theorem}
  There exists an algorithm for the $(m-1)$-committee utilitarian election on the general metric 
  that obtains a 1-distortion of, at most $1 + 2/(m-1)$. Additionally, no algorithm can obtain a 1-distortion better than $1 + 2/(m-1)$ when choosing $m-1$ candidates, even if the metric space is $2$-D Euclidean.
\end{theorem}
Note that the above theorem immediately implies that no algorithm can find a set of size $m-1$ guaranteed to contain an optimum, even if the underlying metric is $2$-D. 
We further study the \emph{max-cost} objective, showing that one can obtain 1-distortions of $1$, $1.5$, and $2$ using sets of size four, three, and two, respectively. 
\begin{theorem}
    For any $k \in \{2, 3, 4\}$, there exists an algorithm for the $k$-committee egalitarian election on the line metric, which obtains a 1-distortion of at most $3 - k/2$.
    Furthermore, there is no algorithm that can obtain a 1-distortion better than $3 - k/2$.
\end{theorem}

For the \emph{max-cost} objective, even though there exists an algorithm that selects one candidate with a distortion of 3, we show that no algorithm can achieve a 1-distortion better than 3, even when choosing a committee of $m-1$ candidates.

\begin{theorem}
  There is no algorithm for the $(m-1)$-committee egalitarian election on $2$-D Euclidean metric such that can obtain a 1-distortion better than $3$.
\end{theorem}

It is important to note that the obtained lower bounds apply to deterministic algorithms.

\begin{table}[h!]
    \centering
    \begin{tabular}{@{}lllcc@{}}
        \toprule
        \textbf{Objective} & \textbf{Metric Space} & \textbf{Committee Size (Out of $m$)} & \textbf{Lower-Bound} & \textbf{Upper-Bound}  \\ 
        \midrule
        \multirow{4}{*}{Sum} & \multirow{2}{*}{1D} & $\ge 2$ & 1 & \textcolor{blue}{1} \\ \cmidrule{3-5}
        && $1$ & 3\cite{anshelevich2018approximating} & 3\cite{gkatzelis2020resolving} \\ \cmidrule{2-5}
        & 2D & $m-1$ & \textcolor{blue}{$1+\frac{2}{m-1}$} & \textcolor{blue}{$1+\frac{2}{m-1}$} \\ 
        \midrule
        \multirow{6}{*}{Max} & \multirow{4}{*}{1D} & $\ge 4$ & 1 & \textcolor{blue}{1} \\ \cmidrule{3-5}
        && $3$ & \textcolor{blue}{1.5} & \textcolor{blue}{1.5} \\ \cmidrule{3-5}
        && $2$ & \textcolor{blue}{2} & \textcolor{blue}{2} \\ \cmidrule{3-5}
        && $1$ & 3 & 3 \\ \cmidrule{2-5}
        & 2D  & $m-1$ & \textcolor{blue}{3} & 3\\ 
        \bottomrule
    \end{tabular}
    \caption{State-of-the-Art; Blue colored bounds are results of this paper
    }
    \label{tab}
\end{table}

Our approach to the problem involves many novel techniques that we believe are of independent interest.
Most notably, for the line metric, we propose an algorithm that essentially finds the order of candidates and voters.
Specifically, we identify the order for
all the candidates in a set of \emph{core candidates}, which has the following property:
If $c_i$ is a core candidate and $c_j$ is not, then all voters prefer $c_i$ to $c_j$.
We additionally find the order of all voters with the minor caveat that we may potentially have ties for voters who have the same exact preference for all core candidates. 
We refer to Section~\ref{order} for more details. 

\subsection{Related works}
\label{sec:related_work}

Since its introduction, the metric distortion framework has remained central in investigating the performance of single-winner voting. Gkatzelis, Halpern, and Shah~\cite{gkatzelis2020resolving} showed the existence of a deterministic mechanism with a 3-distortion factor, settling a long-standing open question, a simpler proof of which was later given by Kizilkaya and Kempe~\cite{kizilkaya2023pluralityvetosimplevoting}. Although a distortion factor of 3 is unavoidable for any deterministic mechanism, an intriguing question arises about whether randomization could surpass this 3-factor barrier. In fact, it was conjectured that a randomized mechanism could achieve a distortion of 2. However, Charikar and Ramakrishnan~\cite{charikar2022metric}, and Pulyassary and Swamy~\cite{pulyassary2021randomized} independently established the non-existence of such a randomized mechanism, refuting the conjecture. The quest of breaking below 3-factor using randomization remained. Charikar \emph{et al.}~\cite{charikar2024breaking} recently answered this positively by developing a randomized mechanism with a 2.753-distortion factor. On a different line of work, Anshelevich \emph{et al.}~\cite{anshelevich2024improved} demonstrate that when the threshold approval set of each voter is known -- containing all candidates whose cost is within an appropriately chosen factor of the voter's cost for their most preferred candidate -- a distortion of \(1 + \sqrt{2}\) can be achieved.

In the $k$-committee election problem (the single-winner election being a special case with $k=1$), the aim is to select $k$ candidates from a pool of $m$ candidates based on ordinal preferences provided by $n$ voters. For any mechanism $f$, its distortion is defined as the worst-case ratio (across all instances) of the cost of the solution produced by $f$ compared to the optimal cost.
When the cost for a voter is considered as the sum of distances to all committee members, Goel, Hulett, and Krishnaswamy~\cite{goel2018relating} showed that the problem reduces to the single-winner election. 

On the other hand, Caragiannis, Shah, and Voudouris~\cite{caragiannis2022metric} considered a general cost function -- each voter's cost is the distance to the $q$-th (for some integer $q \ge 1$) nearest committee member. They identified a trichotomy: For $q \le k/3$, the distortion is unbounded; for $q \in (k/3,k/2]$, it is $\Theta(n)$; and for $q > k/2$, the problem reduces to the single-winner election. 
As an immediate corollary, for the 2-committee election, with each voter's cost being its distance to the nearest committee member (i.e., $q=1$), we get a distortion of $\Theta(n)$.
Further, with the same cost function, for the $k$-committee election when $k \ge 3$, the distortion is unbounded (even for line metric). 
However, when the positions of candidates are known, for $k=m-1$, Chen, Li, and Wang~\cite{chen2020favorite} demonstrated that single-vote rules achieve a distortion of 3 and provided a matching lower bound.

One of the most basic versions -- where both voters and candidates are positioned on a real line -- has already garnered significant attention in computational social choice theory. If we know the locations of the voters and candidates on this real line, a straightforward dynamic programming approach can solve the $k$-committee election problem optimally in $O(nk \log n)$ time. Strict preference profiles, with voters and candidates on a real line (also known as 1-D Euclidean), exhibit many intriguing properties, including being single-peaked and single-crossing~\cite{black1948rationale, mirrlees1971exploration, escoffier2008single}. 
Given a preference profile and the order of voters, deciding whether it is 1-D Euclidean can be done in polynomial time~\cite{elkind2014recognizing}. Furthermore, if the input preference order is consistent with 1-D Euclidean,~\cite{elkind2014recognizing} provides an efficient construction of a mapping realizing that.

Another closely related question to the problem of optimal candidate selection is the facility location problem~\cite{mahdian2006approximation}, where the goal is to place facilities at locations in a metric space to minimize the cost of serving agents. Unlike the candidate selection problem, where candidates are restricted to a fixed set, facilities in the facility location problem can be placed anywhere in the space~\cite{feldman2016voting}.

A related solution concept to the matroid we propose is the Condorcet winning set: a set of candidates such that no other candidate is preferred by at least half the voters over every member of the set~\cite{elkind2015condorcet}. The Condorcet dimension, defined as the minimum cardinality of a Condorcet winning set, is known to be at most logarithmic in the number of candidates. Caragiannis \emph{et al.}~\cite{caragianniscan} partially reaffirm this logarithmic bound when considering committee-selected alternatives. In the metric ranking framework, Lassota \emph{et al.}~\cite{lassota2024condorcet} recently demonstrated that the Condorcet dimension is at most three under the Manhattan or \(\ell_\infty\) norms. Independently and concurrently, Charikar \emph{et al.}~\cite{charikar2024six} showed that Condorcet sets of size six always exist.